\documentclass[12pt]{article}

\usepackage{hyperref}
\hypersetup{
    colorlinks=true,
    linkcolor=lightblue,
    filecolor=magenta,
    urlcolor=lightblue,
    citecolor=vermillion,
}
\usepackage{draftwatermark}
\SetWatermarkText{}

\SetWatermarkScale{.5}
\SetWatermarkColor[gray]{0.9} %
\usepackage{setspace}
\usepackage{geometry}
\usepackage{enumitem}
\usepackage{natbib}
\usepackage{makeidx}
\usepackage{comment}
\usepackage{bbm}
\usepackage[title,titletoc]{appendix}
\usepackage{etoolbox}
\patchcmd{\appendices}{\quad}{: }{}{}
\makeindex
\usepackage{tikz}
\usetikzlibrary{calc}
\usetikzlibrary{arrows.meta}
\usepackage{pgfplots}
\pgfplotsset{compat=1.18}
\usepgfplotslibrary{fillbetween}
\usepackage[hang,bottom,flushmargin]{footmisc}
\usepackage{graphicx}
\usepackage{parskip}
\usepackage{caption}
\usepackage{xcolor}
    \definecolor{lightblue}{RGB}{0,114,178}
    \definecolor{Ggreen}{RGB}{34,177,76}
    \definecolor{Oorange}{RGB}{230,159,0}
    \definecolor{vermillion}{RGB}{213,94,0}
    \definecolor{matgreen}{HTML}{2ca02c}
    \definecolor{matorange}{HTML}{ff7f0e}
    \definecolor{matred}{HTML}{d62728}
    \definecolor{matblue}{HTML}{1f77b4}
    \definecolor{matpurple}{HTML}{9467bd}

    \definecolor{revgreen}{RGB}{0,128,0}

\usepackage{amsmath,amssymb}
\usepackage{amsthm}
\usepackage{apptools}
\usepackage{etoolbox}
\AtAppendix{\counterwithin{theorem}{section}}
\AtAppendix{\counterwithin{lemma}{section}}
\AtAppendix{\counterwithin{corollary}{section}}
\AtAppendix{\counterwithin{proposition}{section}}
\AtAppendix{\counterwithin{condition}{section}}

\newtheorem{theorem}{Theorem}

\newtheorem{corollary}{Corollary}

\newtheorem{lemma}{Lemma}

\newtheorem{proposition}{Proposition}
\newtheorem{remark}[theorem]{Remark}

\newtheorem{definition}[theorem]{Definition}

\AtBeginEnvironment{theorem}{\vskip 6pt}
\AtBeginEnvironment{lemma}{\vskip 6pt}
\AtBeginEnvironment{proposition}{\vskip 6pt}
\AtBeginEnvironment{corollary}{\vskip 6pt}
\AtBeginEnvironment{assumption}{\vskip 6pt}
\AtBeginEnvironment{definition}{\vskip 6pt}
\AtBeginEnvironment{remark}{\vskip 6pt}

\def\equationautorefname~#1\null{Equation~(#1)\null}

\theoremstyle{definition}

\usepackage{newpxtext}
\usepackage{newpxmath}

\newcommand*{\lt}{\left}
\newcommand*{\rt}{\right}

\newcommand{\multi}{\overline{\mathbbm{1}}}
\newcommand{\prior}{\mu_1}
\newcommand{\out}{\theta}
\newcommand{\outI}{\widetilde{\theta}}

\newcommand{\searcher}{searcher}
\newcommand{\Searcher}{Searcher}
\newcommand{\captive}{captive}

\begin{document}

\raggedbottom

\title{Algorithm Transparency and Search Manipulation: Steering vs. Persuasion\footnotemark[1]}

\footnotetext[1]{The authors are grateful for the helpful comments of Guy Aridor, Gary Biglaiser, Adam Dearing, Rick Harbaugh, Maarten C. W. Janssen, Aaron Kolb, Volker Nocke, Marilyn Pease, Marcel Preuss, Andrew Rhodes, Christian Schmid, Ernst-Ludwig von Thadden, Achim Wambach and Cole Williams. Some of the findings appeared in Section 3.2 of the March 2024 version of ``Signaling With Commitment,'' 
available \href{https://arxiv.org/pdf/2305.00777v2}{here}. All remaining errors are solely the authors'.}

\author{\begin{tabular}{ccc}
Raphael Boleslavsky\footnotemark[2] &\quad Thomas Jungbauer\footnotemark[3] &\quad Mehdi Shadmehr\footnotemark[4]
 \end{tabular}}

\footnotetext[2]{Business Economics \& Public Policy, Indiana University: \href{mailto:rabole@iu.edu}{rabole@iu.edu}.}
\footnotetext[3]{Strategy \& Business Economics, Cornell University: \href{mailto:jungbauer@cornell.edu}{jungbauer@cornell.edu}.}
\footnotetext[4]{Public Policy, University of North Carolina at Chapel Hill: \href{mailto:mshadmeh@gmail.com}{mshadmeh@gmail.com}.}

\date{\today}

\maketitle

\setcounter{footnote}{0}
\renewcommand\thefootnote{\arabic{footnote}}

\vskip 30pt

\setstretch{1.1}

\begin{quotation}
\noindent {\small {\sc{Abstract}}}.
 We study a platform that prefers to sell the more profitable of two products. It designs an algorithm that determines the product the consumer encounters first, conditional on her best match. The algorithm simultaneously manipulates consumer attention (steers) and communicates information about match quality (informs). When the algorithm is opaque, it is difficult for a consumer to understand how product order is generated and what it reveals. In the platform’s preferred equilibrium, it places the profitable product first. When the algorithm is transparent, the consumer understands the algorithm and what it conveys. Thus, the algorithm can persuade as well as steer. In some cases, the equilibrium algorithm deters search, in others, encourages it. In the former case, transparency helps consumers; in the latter, it harms some or all of them. Extending or shifting transparency requirements uncouples steering from information provision, reverting the consumer's welfare to opacity.

\end{quotation}

\setstretch{1.3}

\vskip 15pt


\vskip 5pt


\thispagestyle{empty}

\newpage

\setcounter{page}{1}

\textit{Platforms often have a financial incentive to steer customers to particularly profitable products and can use the power of defaults and ordering to accomplish that effectively.}
\begin{flushright} ---\citet{S2019} pg. 51\end{flushright}

\section{Introduction}\label{sec: intro}

Online platforms do not passively host products: by controlling product ordering, they influence consumers’ search incentives, steering them toward certain products and away from others. Indeed, the empirical search literature documents that prominently displayed products attract disproportionate attention, affecting consumers’ search decisions and downstream choices.\footnote{See, e.g., \citet{Ghose_Ipeirotis_Li_2014_MgmtSci}, \citet{Ursu_2018_MktgSci}, \citet{Sahni_Nair_2020_RESTUD}, and \citet{Yu2025WelfareSponsored}.} Because prominence strongly affects consumer choice, regulators and policy observers are particularly concerned that the algorithms which determine these default options are opaque, obscuring the link between a consumer’s attributes and how products are presented to her on the platform. Such opacity makes it difficult for consumers to understand why a particular product is prominent, and what prominence reveals about its suitability.\footnote{``[Platforms] use algorithms to curate personalized content for users based on [an] unprecedented level of data\ldots These algorithms are opaque, and while platforms know exactly which individuals are exposed \ldots and why\ldots the public have very little knowledge'' \citep[p.~150]{S2019}.} Consequently, regulators have pushed for greater transparency, emphasizing that platforms should ``ensure that [users] are appropriately informed about how\ldots systems impact the way information is displayed\ldots and prioritized for them'' \citep[Recital~70]{EU_2022_DSA}.

In this paper, we study how platforms design algorithms to manipulate consumers' search incentives, focusing on the positive and normative consequences of algorithm transparency. By making a product prominent, the platform simultaneously \emph{steers} the consumer toward it and \emph{informs} the consumer about its suitability. In line with regulators' concerns, with an opaque algorithm, the consumer does not exactly know how to interpret prominence. Consequently, the platform cannot leverage prominence to provide information to its advantage. Transparency therefore gives the platform greater control. Broadly, we show that a transparency requirement can induce platforms to alter their algorithms in ways which can benefit or harm consumer welfare. We further show that, if transparency benefits consumers when prominence simultaneously steers and informs, then it does not benefit consumers if these two roles are uncoupled. In particular, if transparency extends beyond an algorithm that governs prominence or ordering, the benefit for consumers is negated. Together, these findings suggest that imposing or strengthening transparency requirements may be counterproductive and should thus be approached cautiously.

Our model combines sequential search with algorithm design. A platform offers two products to a consumer and has an incentive to steer the consumer toward the more profitable one. While the consumer is uncertain which product best matches her needs, the platform can leverage its vast consumer data to predict fit. Before the consumer arrives, the platform designs an algorithm that specifies the probability each product is featured prominently, conditional on the consumer's best match. Once she arrives, the platform observes the consumer's attributes, and implements the algorithm. The consumer then searches sequentially. She first encounters the prominent product and decides whether to buy it immediately or continue. If she continues, she incurs a privately known search cost, encounters the second product, and decides which one to purchase. The consumer faces search and information frictions: an encounter does not always reveal a product's fit, and  her search cost is sometimes prohibitive, compelling her to buy the prominent product. 

In this model, prominence affects the consumer's incentives through two distinct channels. Because search always begins with the prominent product and continuation is costly, the consumer is steered toward the prominently displayed product. At the same time, the algorithm assigns prominence based on the platform's information about product fit with the consumer. Thus, prominence can also signal this information to the consumer. What the consumer infers from prominence, however, depends crucially on whether she understands the platform's algorithm, i.e., whether it is opaque or transparent.

The first part of our analysis focuses on an opaque environment, in which the consumer does not know how the platform customizes prominence based on product fit. She therefore interprets prominence based on her conjecture about the display algorithm, which may differ (out of equilibrium) from the actual algorithm that was implemented by the platform. In particular, a deviation from the consumer's conjecture is not observed, and therefore does not change what she infers from a particular product being prominent. In other words, when the algorithm is opaque, the platform cannot influence how the consumer interprets prominence.  Opacity therefore allows for multiple equilibria.

We characterize all equilibria under opacity in Proposition \ref{prop:equilibrium under opacity}. At every prior belief, equilibria exist in which the more profitable product is always prominent, and the consumer learns nothing about its fit. For some priors, a continuum of equilibria exists in which both products are sometimes prominent, and prominence reveals some information. In all such equilibria, prominence is interpreted unfavorably: a prominent display reduces the consumer's belief that the product is her best match. Consequently, the platform cannot use prominence to steer the consumer as effectively, which reduces its equilibrium profit. For some prior beliefs, equilibria exist in which the consumer's interpretation is so unfavorable that the platform effectively loses its ability to steer. To avoid substantially damaging the consumer's belief about the more profitable product, it is forced to always make the less profitable product prominent, which erodes its profit even further. 

In light of this multiplicity, the platform would benefit from a device that selects its preferred equilibrium, in which the more profitable product is always prominent. One such device is a non-binding norm against using consumer data to customize product ordering or prominence. If the consumer believes that the platform adheres to such a norm, then observing the prominent product does not change her beliefs about product fit. Thus, the platform is free to make the more profitable product prominent with no adverse consequences, and it does so to steer the consumer as aggressively as possible.

 The platform’s incentives are starkly different in a transparent environment, where the consumer understands how prominence is determined by the algorithm and what it conveys about fit. Because the consumer understands how prominence is assigned and what it means, the platform directly controls the consumer's interpretation of prominence. Beyond its effect on the search order, the algorithm can therefore be used to convey information which persuades the consumer to adopt a more profitable  search strategy.

 Proposition \ref{prop:equilibrium under transparency} characterizes the equilibrium under transparency. The equilibrium algorithm takes one of three forms, depending on the prior belief. When the more profitable product is very likely to be the consumer’s best match, the algorithm always makes it prominent. In this case, the interests of the platform and the consumer are sufficiently aligned that the platform has no need to convey additional information. When the consumer is less optimistic about the more profitable product, however, the platform benefits from designing an algorithm that persuades the consumer to alter her search strategy. When the prior is moderate, prominence conveys good news, which makes the consumer more inclined to stop searching and purchase the prominent product. When the prior is low, however, prominence conveys bad news, which makes the consumer more inclined to continue searching. Thus, with a transparent algorithm, in equilibrium prominence  may convey good news that deters further search, or bad news that encourages it.

The interplay between steering and persuasion has significant welfare implications. At moderate priors, prominence deters search, steering the consumer toward a product that she would like to purchase immediately. Thus, steering and persuasion reinforce each other to the consumer's benefit. At low priors, the two effects instead work in opposite directions: prominence incentivizes continued search by steering the consumer away from the product that she infers is likely to be the better match. Proposition \ref{prop:transparency and welfare} characterizes the welfare consequences. Relative to the platform’s preferred equilibrium under opacity, transparency increases the consumer’s welfare at moderate priors. At low priors, however, it reduces her welfare whenever her search cost is prohibitive; when search is viable, transparency may reduce or increase her welfare, depending on the combination of search and information frictions that she faces. 

The preceding analysis focuses on algorithms that determine prominence; if the platform wishes to convey information to the consumer about product fit, it can only do so by manipulating her search order. In Section \ref{sec:wrong kind of transparency}, we decouple information provision from prominence. The platform designs an algorithm that sends the consumer a (payoff-irrelevant) product recommendation, conditional on her best match. If this algorithm is opaque, then recommendations are cheap talk. If transparent, however, recommendations displace prominence as tools of persuasion. In Proposition \ref{prop:equilibrium joint design}, we show that recommendations convey information in equilibrium, but prominence is always assigned to the more profitable product and thus conveys no additional information about fit. Furthermore, recommendations never strictly overturn the consumer's prior search strategy. Thus, the consumer's equilibrium payoff is as if the more profitable product is always prominent, exactly as in the platform-preferred equilibrium under opacity (Proposition \ref{prop:wrong kind of transparency}). For moderate priors, transparency therefore helps the consumer if it applies to the algorithm that determines prominence, but not if it applies to the algorithm that recommends products.

In Section \ref{sec:extensions} we extend our analysis in two directions. First, we study how the platform can monetize prominence, endogenizing the commission that the platform earns on sales of each product. Transparency not only helps the platform to boost sales of the product with higher commissions, but also increases the expected commissions and the platform's profit. Second, we study how transparency affects the mix of products offered by the platform. We show that the platform may benefit by sometimes prominently displaying a product that is always worse for both the consumer and platform than another product: by doing so, the platform gains additional influence over the consumer's search.

\textbf{Related Literature.}  Our paper is built on two main insights. First, product prominence has two effects: it \emph{steers} by constraining the search order, and it \emph{informs} by signaling the platform's private information about match quality. Second, if transparency requirements achieve their goal, so that the consumer understands how prominence is determined and what it reveals about product fit, then prominence can be used to shape the consumer's inferences, becoming an instrument for persuasion. Related ideas appear in the existing literature on consumer search, platform intermediation, and transparent communication.

 A strand of literature examines how fixed prominence or search order affect market outcomes. \citet{A2007} studies competition with a predetermined search order, \citet{AVZ2009} with a prominent firm that is visited first by all consumers. In these papers order is fixed, and conveys no useful information to consumers. Another strand endogenizes prominence through a bidding process \citep{Armstrong_Zhou_2011_EJ,BarIsaac_Shelegia_2025_JPEMicro}. In these papers, firms are uninformed, and prominence does not convey information about fit. \citet{Hagiu_Jullien_2011_RAND} study a model that is  similar to ours, except that consumers always learn their types perfectly at the first encounter (i.e., there is no information friction). Consequently, consumers do not learn from prominence: it steers but does not inform. 

 A number of related strands of literature study ranking, prominence, or order purely as a signal. Thus, product ranking reveals information, but it does not constrain consumers' subsequent search decisions or affect inspection costs: it informs but does not steer. In sponsored-search models, bids and positions reveal seller relevance, directing a subsequent (unconstrained) sequential search \citep{Athey_Ellison_2011_QJE, CH2011}.  \citet{Nocke_Rey_2024_JPE} consider a related model of a multiproduct firm, showing that it may prefer an equilibrium in which product order is uninformative to one in which it is revealing. \citet{Janssen_et_al_2026_EJ} allow a platform to assign sponsored and organic positions using bids and private match assessments, while \citet{Janssen_et_al_2026_IJIO} characterize optimal ranking informativeness. In \citet{BB2024} advertising both steers and informs: competing firms bid to target an ad to consumers, who can subsequently buy the product off-platform. The authors characterize the revenue maximizing mechanisms for ad sales. Another strand of literature studies intermediaries who advise consumers directly, rather than via rankings \citep{IO2009,deCorniere_Taylor_2019_RAND,Teh_Wright_2022_AEJMicro}.

Methodologically, our paper is related to the literature on transparency in communication games. \citet{Kamenica_Gentzkow_2011_AER} study a communication game with payoff-irrelevant messages, in which the sender's strategy is observed and understood by the receiver, endowing the sender with commitment power.  In our model, in contrast, prominence has direct payoff consequences. Thus, our approach is closer to \citet{BS2025}, which studies commitment in signaling games.

\newpage
\section{Model}\label{sec:model}

\paragraph{Market.} Consider a platform selling two products, $A$ and $B$, to a single representative consumer. The consumer is familiar with product $B$, which yields a known payoff $\out\in(0,1)$ if it is purchased. Product $A$ is new to the consumer, yielding uncertain payoff $\omega\in\{0,1\}$ if purchased. Thus, $\omega=1$ ($\omega=0$) implies that $A$ is a better (worse) match for the consumer than $B$.  The common prior that $A$ is the better match is denoted by $\prior \equiv \Pr(\omega=1)$, whereas $\mu\equiv\Pr(\omega=1)$ denotes an arbitrary belief. The consumer searches for a product sequentially, as described below. She maximizes the payoff from her chosen product, net of search cost. The platform is compensated through commissions, earning $k_i$ when product $i\in\{A,B\}$ is sold. In our baseline analysis, product $A$ is more profitable for the platform, $k_A>k_B$, and we normalize $k_A=1$ and $k_B=0$. As described in more detail below, the platform acts ahead of the consumer's search, maximizing the probability that the consumer chooses product $A$.\footnote{A trivial way for the platform to maximize sales of $A$ is not to carry $B$. If it did so, the platform would lose the commissions paid by $B$ and upset $A$'s incentive to offer commissions in the first place.}

\paragraph{Learning.} In order to purchase a product, the consumer must encounter it on the platform. Whenever she encounters product $A$, the consumer sees ratings, reviews, descriptions or other information that allow her to learn about her payoff. In particular, when $A$ is encountered, the consumer observes a truth-or-noise signal $R\in\{0,1,\phi\}$ satisfying
\begin{align*}
\Pr(R=0\mid \omega=0)=r,\qquad
\Pr(R=1\mid \omega=1)=r\quad\text{and}\quad
\Pr(R=\phi\mid \omega)=1-r,
\end{align*}
 where $r\in(0,1)$. Thus, with probability $r$ the signal reveals the state perfectly and with probability $1-r$ it is inconclusive, revealing no new information. We adopt the convention that realizations $R=1$ and $R=0$ reveal the state, even if the consumer is surprised.

\paragraph{Timing and Transparency.}
The game unfolds over four stages.
\begin{itemize}
\item[(i)] \textit{Design}: The platform designs a display algorithm, $\Pi\equiv(\pi_0,\pi_1)$, where $\pi_\omega$ is the probability that product $A$ is  prominent in state $\omega$.

\item[(ii)] \textit{Arrival}: The consumer arrives. With transparency, she observes the algorithm $\Pi$; with opacity, she does not.\footnote{``Observing the algorithm'' serves a proxy for the consumer understanding how prominence is assigned.} The platform privately observes the state and implements the algorithm. The consumer observes which product is prominent.

\item[(iii)] \textit{Prominent Encounter}: The consumer encounters the prominent product. If it is $A$, she privately observes the realization of  signal $R$. The consumer then decides whether to stop or continue. If she stops, then she purchases the prominent product and the game ends. If she continues, then she incurs a privately known search cost $C$, moving to stage (iv).
\item[(iv)] \textit{Second Encounter}: The consumer encounters the second product. If it is $A$, she privately observes the realization of signal $R$. The consumer purchases either product $A$ or $B$ with free recall, and the game ends.
\end{itemize}

\paragraph{Search Cost.} Our model features two types of consumers, \emph{\captive{}s} and \emph{\searcher{}s}. With probability $h\in(0,1)$, the consumer is a \captive{}, with search cost $C=\bar{c}$. We focus on the case where $\bar{c}$ is prohibitively large, so that a \captive{} always purchases the prominent product (see condition (i) below). With probability $1-h$, the consumer is a \searcher{} with $C=c>0$. For the \searcher{}, it may be optimal to pay the cost to consider the second product, depending on her belief about the state. Let $f$ denote the likelihood ratio of \captive{}s to \searcher{}s,
\[
f\equiv \frac{h}{1-h}.
\]
Thus, there are two dimensions of search friction: the probability of being captive, and the \searcher{}'s cost. Note that the search cost is paid to continue past the prominent product.

\paragraph{Regularity Conditions.} Throughout the analysis, we impose the following parametric restrictions,
\begin{itemize}
 \item[(i)] $c<\min\{\out,1-\out\}$ and $\bar{c}>\max\{\out,1-\out\}$,
 \item[(ii)] $f<1-r$,
 \item[(iii)] $r\out(1-\out)<c$.
 \end{itemize}
 Under condition (i), search is prohibitively expensive for a \captive{}. Simultaneously, a fully informed \searcher{} is always willing to incur the cost to purchase the better match.\footnote{If $A$ is prominent and $\omega=0$, the consumer obtains 0 from $A$ and $\out-C$ from $B$. If $B$ is prominent and $\omega=1$, the consumer obtains $\out$ from $B$ and $1-C$ from $A$. With condition (i), the \searcher{} ($C=c)$ selects $B$ in the former case and $A$ in the latter, and vice versa for  a \captive{} ($C=\bar{c}$).} Condition (ii) limits the probability of a \captive{} consumer, focusing the analysis on the most interesting case.  Condition (iii) tempers the incentive to experiment: when $B$ is prominent, a consumer who continues searching buys $A$ if she draws an inconclusive realization.

\paragraph{Discussion.} We briefly discuss aspects of the model.

\noindent\textit{Platform Information.} In our setup, the platform observes the consumer's match value for the new product, but not her search cost. Match quality is a persistent object that the platform can infer from its data, whereas search cost is affected by the consumer's attention, time constraints, or cognitive load, making it transient and difficult to observe.

\noindent\textit{Prominence.} Our notion of prominence is similar to \citet{AVZ2009}, whereby the consumer searches for products sequentially, but always encounters the prominent product first. The consumer can therefore avoid the search cost $C$ by purchasing the prominent product. Thus, prominence directly affects the consumer's search incentives, steering her toward the product. Prominence can also affect her beliefs. In particular, if the observed (or conjectured) display algorithm assigns prominence differently across states, then observing the prominent product also conveys information about it.

\section{Preliminaries}\label{sec:analysis}

\subsection{Search and prominence}\label{subsec:search}
We begin by characterizing the consumer's search strategy, revealing how search frictions, information frictions, and prominence interact. As described previously, the \captive{} consumer always selects the prominent product, as search is prohibitively expensive. Suppose, therefore, that the consumer is a \searcher{} and she holds some belief, $\mu=\Pr(\omega=1)$ after her encounter with the prominent product in stage (iii). If $A$ is prominent and the signal realization is conclusive, then this belief is $0$ or $1$.

\emph{$A$ is prominent}. If she stops, the \searcher{} expects payoff $\mu$; if she continues, $\out-c$. It is optimal for her to stop when $\mu>\theta_A=\out-c$, continue when $\mu<\theta_A$, and mix otherwise. 

\emph{$B$ is prominent}. If she stops, she expects payoff $\out$; if she continues, 
\begin{align*}
U^{con}_B(\mu)=r\lt[\mu +(1-\mu)\out\rt]+(1-r)\max\{\out,\mu\}-c.
\end{align*}
Indeed, continuing allows the \searcher{} to learn the state with probability $r$. If the signal realization is conclusive, then she buys the product with higher payoff; otherwise, she buys $A$ for $\mu\geq \out$ and $B$ otherwise. As such, $U^{con}_B(\cdot)$ is piecewise linear and increasing, with an upward kink at $\mu=\out$. Furthermore, under regularity condition (iii), we have $U^{con}_B(\out)=\out^\dag\equiv \out+r\out(1-\out)-c<\out$. By implication, the \searcher{} is indifferent at some belief $\theta_B>\out$. Thus, if $\mu<\theta_B$ the \searcher{} immediately buys $B$. If $\mu>\theta_B>\out$, then she continues searching, purchasing $A$ unless she discovers that $B$ is a better match ($R=0$).

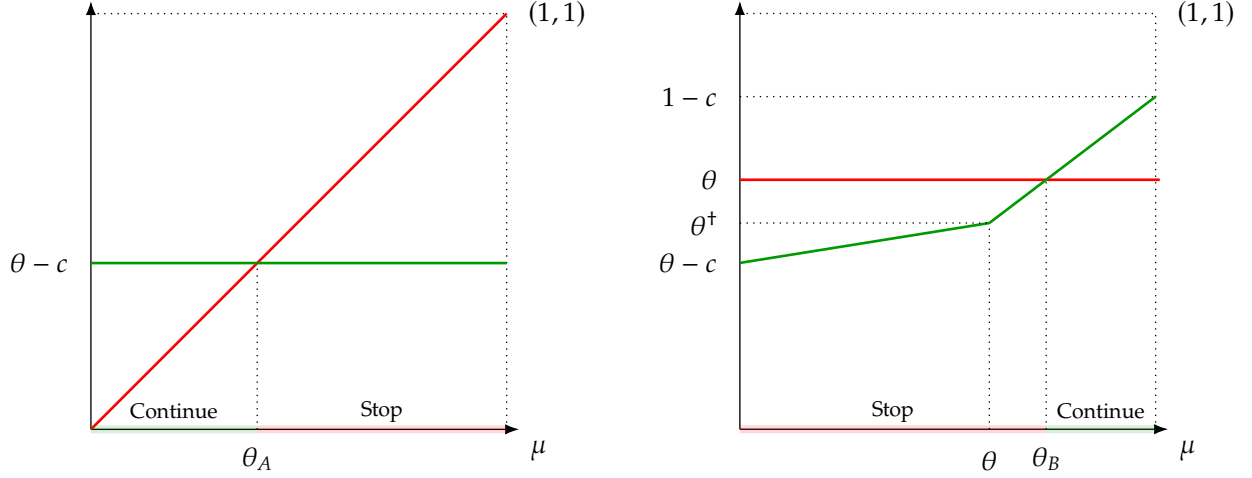
\begin{figure}
\centering

\begin{minipage}{0.48\textwidth}
\centering
\begin{tikzpicture}[
    scale=5.5,
    >=Latex,
    every node/.style={font=\footnotesize}
]

\def\thetaA{0.4} %
\def\u{0.6}

\draw[->, line width=0.5pt] (0,0) -- (1.03,0)
    node[below right] {{$\mu$}};

\draw[->, line width=0.5pt] (0,0) -- (0,1.03);

\draw[dotted, line width=0.5pt]
    (0,1) -- (1,1);

 \draw[dotted, line width=0.5pt]
     (1,0) -- (1,1);

\draw[red, line width=1pt] (0,0) -- (1,1);

\draw[green!60!black, line width=1pt] (0,\thetaA) -- (1,\thetaA);

\draw[dotted, line width=0.5pt]
    (\thetaA,0) -- (\thetaA,\thetaA);

\node[left=4pt] at (0,\thetaA) {\footnotesize{$\out-c$}};
\node[below=4pt] at (\thetaA,0.01) {{$\theta_A$}};
 \node[right=4pt] at (1,1) {{$(1,1)$}};
 \draw[line width=3pt, green!60!black, opacity=0.15] (0,0) -- (\thetaA,0);
 \draw[line width=3pt, red, opacity=0.15] (1,0) -- (\thetaA,0);

\node[above=0pt] at (0.5*\thetaA,0) {\scriptsize{Continue}};

\node[above=0pt] at (0.5+0.5*\thetaA,-0.01) {\scriptsize{Stop}};

\end{tikzpicture}
\end{minipage}
\hfill
\begin{minipage}{0.48\textwidth}
\centering
\begin{tikzpicture}[
    scale=5.5,
    >=Latex,
    every node/.style={font=\footnotesize}
]

\pgfmathsetmacro{\u}{0.6}
\pgfmathsetmacro{\r}{0.4}
\pgfmathsetmacro{\c}{0.2}

\pgfmathsetmacro{\uc}{\u-\c}
\pgfmathsetmacro{\ymid}{\u + \r*\u*(1-\u) - \c}
\pgfmathsetmacro{\yr}{1-\c}
\pgfmathsetmacro{\thetaB}{(\u*(1-\r)+\c)/(1-\r*\u)}

\draw[->, line width=0.5pt] (0,0) -- (1.03,0)
    node[below right] {{$\mu$}};

\draw[->, line width=0.5pt] (0,0) -- (0,1.03);

\draw[dotted, line width=0.5pt]    (0,1) -- (1,1);

\draw[dotted, line width=0.5pt]    (1,0) -- (1,1);

\draw[red, line width=1pt] (0,\u) -- (1.01,\u);

\draw[green!60!black, line width=1pt]
    (0,\uc) -- (\u,\ymid);

\draw[green!60!black, line width=1pt]
    (\u,\ymid) -- (1,\yr);

\draw[dotted, line width=0.5pt]    (0,\ymid) -- (\u,\ymid);

\draw[dotted, line width=0.5pt]    (\u,0) -- (\u,\ymid);

\draw[dotted, line width=0.5pt]
    (\thetaB,0) -- (\thetaB,\u);

\node[left=4pt] at (0,\u) {{$\out$}};
\node[left=4pt] at (0,\ymid) {{$\out^\dagger$}};
\node[left=4pt] at (0,\uc) {{$\out-c$}};

\node[below=4pt] at (\u,0) {{$\out$}};
\node[below=4pt] at (\thetaB,0.01) {{$\theta_B$}};

\node[right=4pt] at (1,1) {{$(1,1)$}};
\node[left=4pt] at (0,\yr) {{$1-c$}};
\draw[dotted, line width=0.5pt](0,\yr)--(1,\yr);

\node[above=0pt] at (0.5*\thetaB,-0.01) 
{\scriptsize{Stop}};
\node[above=0pt] at (0.5+0.5*\thetaB,0) {\scriptsize{Continue}};

\draw[line width=3pt, red, opacity=0.15] (0,0) -- (\thetaB,0);
 \draw[line width=3pt, green!60!black, opacity=0.15] (1,0) -- (\thetaB,0);

\end{tikzpicture}
\end{minipage}

\caption{\textit{Lemma \ref{lem:searcher strategy} (\Searcher{} Strategy)}. Each panel depicts the payoff of stopping (red) and continuing (green) when  \searcher{}'s belief is $\mu$ in stage (iii). Left panel: $A$ is prominent. The \searcher{} continues at $\mu<\theta_A$, always buying $B$; the \searcher{} stops at $\mu>\theta_A$, always buying $A$. The \searcher{} mixes between these at $\mu=\theta_A$.  Right panel: $B$ is prominent. The \searcher{} stops at $\mu<\theta_B$, buying $B$; the \searcher{} continues at $\mu>\theta_B$, buying $A$ iff $R\in\{\phi,1\}$. The \searcher{} mixes between these at $\mu=\theta_B$. }
\label{fig:searcher strategy}
\end{figure}

We have proved the following lemma, which is illustrated graphically in Figures \ref{fig:searcher strategy} and \ref{fig:searcher strategy continued}.

\begin{lemma}[\textbf{\Searcher{} Strategy}]\label{lem:searcher strategy} Consider a \searcher{} deciding whether to purchase the prominent product or continue, with belief $\mu=\Pr(\omega=1)$ after her encounter with the prominent product in stage (iii). The \searcher{}'s optimal strategy is as follows:
\begin{itemize}
\item Suppose $A$ is prominent. (i) If the \searcher{} continues, then she always buys $B$; if she stops, then she always buys $A$. (ii) If the \searcher{} learned that $A$ is a match $(R=1)$ then she stops. (iii) If the \searcher{} learned that $B$ is a match ($R=0$), then she continues. (iv) If the signal realization is inconclusive ($R=\phi$), then the \searcher{} stops if $\mu>\theta_A$, continues if $\mu<\theta_A$, and mixes if $\mu=\theta_A$.
\item Suppose $B$ is prominent. (i) If the \searcher{} continues, then she buys $A$ if and only if it is revealed to be a match or the signal realization is inconclusive ($R\in\{1,\phi\}$); if she stops, then she buys $B$. (ii) The \searcher{} stops if $\mu<\theta_B$, continues if $\mu>\theta_B$, and mixes if $\mu=\theta_B$.
\end{itemize}
Furthermore, $\theta_A=\out-c$ and $\theta_B=\tfrac{\out(1-r)+c}{1-r\out}$. In addition, the thresholds are ordered $0<\theta_A<\out<\theta_B<1$, threshold $\theta_A$ is decreasing in $c$, and $\theta_B$ is increasing in $c$.

\end{lemma}

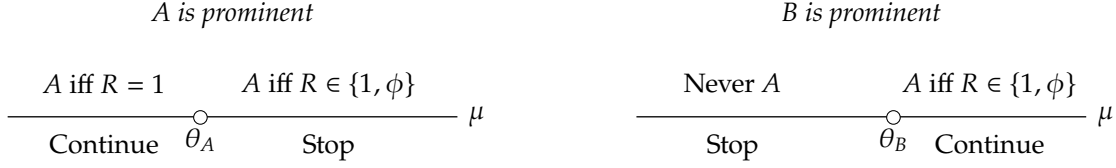
\begin{figure}
\begin{minipage}{0.5\textwidth}
\begin{tikzpicture}[scale=0.85,  >=Latex,]
\pgfmathsetmacro{\off}{4}

\draw[line width=0.5pt] (0,0) -- (7,0);
\fill (7,0) node[right] {\footnotesize{$\mu$}};
\fill (3,0) node[below] {\footnotesize{$\theta_A$}};
\filldraw[color=black, fill=white] (3,-0.01) circle (3pt);
\node[above=0pt] at (3.5,1.2) 
{\footnotesize{$A$ \textit{is prominent}}};

\node[above=0pt] at (5,0.1) 
{\footnotesize{$A$ iff $R\in\{1,\phi\}$}};
\node[below=0pt] at (5,-0.1) 
{\footnotesize{Stop}};

\node[above=0pt] at (1.5,0.2) 
{\footnotesize{$A$ iff $R=1$}};
\node[below=0pt] at (1.5,-0.1) 
{\footnotesize{Continue}};

\end{tikzpicture}
\end{minipage}
\begin{minipage}{0.5\textwidth}
\begin{tikzpicture}[scale=0.85,  >=Latex,]
\pgfmathsetmacro{\off}{4}

\draw[line width=0.5pt] (0,-\off) -- (7,-\off);
\fill (7,-\off) node[right] {\footnotesize{$\mu$}};
\fill (4,-\off) node[below] {\footnotesize{$\theta_B$}};
\filldraw[color=black, fill=white] (4,-0.01-\off) circle (3pt);
\node[above=0pt] at (3.5,1.2-\off) 
{\footnotesize{$B$ \textit{is prominent}}};

\node[above=0pt] at (5.5,0.1-\off) 
{\footnotesize{$A$ iff $R\in\{1,\phi\}$}};
\node[below=0pt] at (5.5,-0.1-\off) 
{\footnotesize{Continue}};

\node[above=0pt] at (1.5,0.2-\off) 
{\footnotesize{Never $A$}};
\node[below=0pt] at (1.5,-0.1-\off) 
{\footnotesize{Stop}};

\end{tikzpicture}
\end{minipage}

\caption{\label{fig:searcher strategy continued} \textit{Lemma \ref{lem:searcher strategy} (Searcher Strategy)} The \searcher{}'s stopping and purchase decisions.}
\end{figure}

\subsection{Belief-based approach}\label{sec:belief-based approach}

We analyze the platform's design problem using the belief-based approach, whereby an algorithm is represented by the posterior beliefs that it induces.

To see how this works, consider an algorithm $\Pi=(\pi_0,\pi_1)$ which determines the probability that each product is displayed prominently in each state $\omega\in\{0,1\}$. Consequently, the algorithm determines the consumer's posterior beliefs at the arrival stage when product $i\in\{A,B\}$ is prominent, $\mu_i=\Pr(\omega=1|i\text{ prominent})$. Simultaneously, the algorithm determines the probability $\tau_i$ with which $i\in\{A,B\}$ is prominent, inducing belief $\mu_i$. In particular, for any algorithm $\Pi$,
\begin{align}\label{eq:BBA}
\tau_A=\prior\pi_1+(1-\prior)\pi_0,\quad\tau_B=1-\tau_A,\quad\mu_A=\frac{\prior\pi_1}{\tau_A},\quad\mu_B=\frac{\prior(1-\pi_1)}{\tau_B},
\end{align}
and the law of iterated expectations implies $\mathbb{E}_\tau[\mu_i]=\tau_A\mu_A+\tau_B\mu_B=\prior$.

As is well known, these steps can also be reversed: any combination of beliefs $(\mu_A,\mu_B)$ and probabilities $(\tau_A,\tau_B)$ that satisfies the law of iterated expectations is induced by some display algorithm $\Pi$.\footnote{It is straightforward to compute that $\pi_1=\tau_A\mu_A/\prior$ and $\pi_0=\tau_A(1-\mu_A)/(1-\prior)$.}
 Given only beliefs $(\mu_A,\mu_B)$, such a display algorithm exists if and only if the prior lies in between them ($\prior\in[\min\{\mu_A,\mu_B\},\max\{\mu_A,\mu_B\}]$), with the associated probabilities $(\tau_A,\tau_B)$ determined by the law of iterated expectations. 

Whether and what type of information an algorithm conveys is determined by the relationship between the posteriors and the prior. If the posterior $\mu_i=\prior$ for one or both $i\in\{A,B\}$, then prominence is \textit{uninformative}, as the realized posterior belief is always equal to the prior.\footnote{If $\mu_i=\prior$ and $\mu_j\neq\prior$, then \eqref{eq:BBA} requires $\tau_i=1$ and $\tau_j=0$. With $\mu_A=\mu_B=\prior$, the posterior is equal to the prior, regardless of which product is prominent.} If instead $\mu_A\neq\prior$ and $\mu_B\neq\prior$, then the posterior belief always differs from the prior, and prominence is \textit{informative}. Furthermore, when $\mu_A<\prior<\mu_B$ a prominent display reduces the consumer's belief that the product is her best match; thus, \textit{prominence conveys bad news}. Similarly, for $\mu_B<\prior<\mu_A$, \textit{prominence conveys good news}.

\section{Equilibrium Algorithms}\label{sec:equilibrium algorithms}

We characterize the equilibrium display algorithms under opacity and transparency and discuss their properties. Throughout, we focus on strong PBE: (1) the consumer follows the optimal search strategy given her stage (iii) belief (see Lemma \ref{lem:searcher strategy}), (2) platform algorithm is optimal given the consumer's search strategy, (3) beliefs are updated with Bayes' rule when possible.\footnote{With opacity, the consumer is surprised whenever she conjectures that one product is always prominent, but the other product is observed. We characterize all equilibria, placing no restriction on the off-path belief. With transparency, the consumer is never surprised.}

\subsection{Opacity}\label{sec:opacity}

We begin with the opaque case, in which the consumer does not directly observe the display algorithm. Instead, she observes only the prominent product and updates beliefs based on her conjecture about the algorithm. While the platform can use the algorithm to assign prominence allowing it to exploit the search frictions, it cannot directly affect how the consumer interprets prominence, creating the possibility of multiple equilibria.

\paragraph{Opaque Profit.} Because the consumer does not observe the platform's true algorithm, $\Pi$, her beliefs at the arrival stage and subsequent search behavior depend on a conjectured algorithm, $\widehat{\Pi}$, which induces posterior beliefs $(\widehat{\mu}_A,\widehat{\mu}_B)$ at the end of the arrival stage, when the prominent product is observed. These beliefs, in turn, determine the consumer's search strategy. Because the true algorithm $\Pi=(\pi_0,\pi_1)$ is unobserved, it determines the probability with which each product is prominent in each state, but it does not affect the meaning of a prominent display, as reflected in the consumer's beliefs $(\widehat{\mu}_A,\widehat{\mu}_B)$. Thus, the consumer's search strategy depends on her beliefs, but not on the algorithm itself.

Let $a$ denote the probability that the \searcher{} stops when $A$ is prominent if her signal realization is inconclusive ($R=\phi$), and let $b$ denote the probability that she continues when $B$ is prominent and her signal realization is inconclusive $(R=\phi)$.  From Lemma \ref{lem:searcher strategy} (consult Figure \ref{fig:search}), in equilibrium we must have 
\begin{equation*}a\in\multi\{\widehat{\mu}_A\geq\theta_A\}\qquad \qquad b\in\multi\{\widehat{\mu}_B\geq\theta_B\},\end{equation*}
where $\multi\{x\geq y\}$ denotes the indicator correspondence: it coincides with the usual indicator function when the inequality inside is strict, and it is \([0,1]\) at equality. In other words, this correspondence fills in the jump from \(0\) to \(1\) at equality, yielding a closed graph. Throughout, inequalities between sets must hold for all possible selections.\footnote{In particular, given compact sets $X\subset\mathbb{R}$ and $Y\subset\mathbb{R}$, we say that $X>Y$ if and only if $\min X>\max Y$.}

Normalizing by $1-h$, the platform's profit from algorithm  $\Pi$ when $\widehat{\Pi}$ is conjectured is 
\begin{equation}\label{eq:opaque profit}
\prior\bigl[\pi_1\underbrace{\{f+r+(1-r)a\}}_{z_A^1}+(1-\pi_1)\underbrace{b}_{z_B^1}]+
(1-\prior)\bigl[\pi_0\underbrace{\{f+(1-r)a\}}_{z_A^0}+(1-\pi_0)\underbrace{(1-r)b}_{z_B^0}\bigr].
\end{equation}
This expression also build on Lemma \ref{lem:searcher strategy}. Each term $z_i^\omega$ is the platform's expected profit conditional on $i\in\{A,B\}$ prominent and state $\omega\in\{0,1\}$. First, $z_A^1=f+r+(1-r)a$: the \captive{} always buys $A$, the \searcher{} buys $A$ when the signal is conclusive, and also when it is inconclusive and she stops. Second, $z_B^1=b$: only the \searcher{} buys $A$, and she does so whenever she continues.\footnote{Because $\omega=1$, the signal realization is either $R=1$ or $R=\phi$, both of which lead to a purchase of $A$.} Third, $z_A^0=f+(1-r)a$: relative to $z_A^1$, the \searcher{} no longer buys $A$ when the signal is conclusive, as it reveals $\omega=0$. Finally, $z_B^0=(1-r)b$: the \searcher{} buys $A$ only when she continues and the signal is inconclusive.

\begin{remark}(State Dependence). Although the platform only cares about the consumer's final purchase, its expected payoff is nevertheless state-dependent, as is clear from \eqref{eq:opaque profit}. State-dependence arises from learning, which impacts the \searcher{}'s product choice differently across states. If the information friction is maximal ($r=0$) then this state-dependence disappears, and other equilibria are also possible.
\end{remark}

\paragraph{Opaque Algorithm Design.} From the expression for platform profit \eqref{eq:opaque profit}, the platform's optimal algorithm $\Pi$ when $\widehat{\Pi}$ is conjectured, is 
\[\pi_1=\multi\{z_A^1\geq z_B^1\}\qquad\qquad \pi_0=\multi\{z_A^0\geq z_B^0\}.\] In each state, the platform chooses the prominent product to maximizes its expected profit, mixing if indifferent.

\paragraph{Opaque Equilibrium.} We describe equilibria using the following taxonomy.

\textit{$i$-always equilibrium}. In an $i$-always equilibrium, product $i\in\{A,B\}$ is always prominent ($\tau_i=1$). Any such equilibrium is uninformative ($\mu_i=\prior$). All equilibria in this class differ only in the off-path belief $\mu_j$ ($j\neq i$) and are payoff-equivalent.

\textit{Split equilibrium}. In a split equilibrium, the algorithm is informative ($\mu_A\neq\prior$ and $\mu_B\neq\prior$) and, therefore, both products are sometimes prominent ($\tau_A,\tau_B\in(0,1)$). All equilibria in which the algorithm is informative fall into this class, including standard separating and semi-separating strategies.

In the Appendix, we prove the following.

\begin{proposition}[\textbf{Equilibrium Under Opacity}]\label{prop:equilibrium under opacity} When the display algorithm is opaque, 
\begin{itemize}
\item[(i)] $A$-always equilibria exist at every prior $\prior\in(0,1)$. 
\item[(ii)] $B$-always equilibria exist if and only if $\prior\geq\theta_B$. At $\prior=\theta_B$, the searcher's mixing must be such that $b\geq f+r$.
\item[(iii)] split equilibria exist if and only if $\prior>\theta_A$. In any such equilibrium, prominence  conveys bad news ($\mu_A<\prior<\mu_B$). 
\item[(iv)] whenever split or $B$-always equilibria exists ($\prior\geq\theta_A$), the platform strictly prefers the $A$-always equilibria.

\end{itemize}
\end{proposition}

A partial intuition can be gleaned from Figure \ref{fig:equilibrium under opacity}. The left panel illustrates the dependence of $z_A^1$ and $z_B^1$ on the associated posterior belief, accounting fully for the consumer's optimal search strategy. Thus, the left graph illustrates the correspondences $\overline{z}_i^1(\cdot)$ ($i\in\{A,B\}$), obtained by replacing the mixing probabilities ($a$ or $b)$ in $z_i^1$ with the entire set of possible values, $\multi\{\mu\geq \theta_i\}$. The right panel illustrates the analogous objects for state 0, i.e., $\overline{z}_i^0(\cdot)$. The vertical segments arise from the \searcher{}'s mixing at the threshold beliefs, and a selection is determined by her mixed strategy.\footnote{In an equilibrium with mixing, selection must be consistent with the \searcher{}'s strategy in both states. In other words, at $\mu_A=\theta_A$ ($\mu_B=\theta_B$) the payoff selections cannot be made independently in both graphs, as they must be consistent with a specific mixing probability $a\in[0,1]$ ($b\in[0,1]$).} To simplify the explanation, we abstract from mixing. Whether the platform wishes to make $A$ prominent or $B$ prominent in state $\omega$ is determined by the ranking of $\overline{z}^\omega_A(\mu_A)$ and $\overline{z}^\omega_B(\mu_B)$.

A few simple observations connect to Proposition \ref{prop:equilibrium under opacity}. First note that in both states $\bar{z}_A^\omega(\mu)>\bar{z}_B^\omega(\mu)$. In other words, holding the belief fixed, making $A$ prominent is preferred in both states. This observation reflects the steering benefit: when $A$ is prominent the platform sells to all \captive{}s (shifting the normalized profit up by $f$) and the \searcher{} is more inclined to stop after encountering it ($\theta_A<\theta_B$). Thus, it is intuitive that $A$-always equilibria exist at all priors. Indeed, in such an equilibrium $\mu_A=\prior$, while $\mu_B$ is off-path and can be chosen freely. Setting $\mu_B=\prior$ (for example), ensures $\overline{z}_A^\omega(\prior)>\overline{z}_B^\omega(\prior)$, so making $A$ prominent in both states is indeed optimal for the platform. Second, note that the ranges of these correspondences overlap (recall $0<f<1-r$). This overlap also allows for $B$-always equilibria. Indeed, in such equilibria $\mu_B=\prior>\theta_B$. A sufficiently pessimistic off-path belief $\mu_A<\theta_A$ ensures $\overline{z}_B^\omega(\prior)>\overline{z}_A^\omega(\mu_A)$, thereby sustaining the $B$-always equilibria. While the characterization of split equilibria is presented in the Appendix, the key point that prominence reveals bad news can also be intuited from the figure.\footnote{In the Appendix, we show that in any split equilibrium, prominence is not only bad news ($\mu_A<\prior<\mu_B$), but also $\mu_A=\theta_A$. Thus, such equilibria only exist when $\prior>\theta_A$. } Imagine a possible split equilibrium in which prominence reveals good news, $\mu_B<\prior<\mu_A$. Because the graphs are increasing, $\overline{z}_B^\omega(\mu_B)<\overline{z}_B^\omega(\prior)<\overline{z}_A^\omega(\prior)<\overline{z}_A^\omega(\mu_A)$. Consequently, $\overline{z}_A^\omega(\mu_A)>\overline{z}_B^\omega(\mu_B)$, so the platform would make $A$ prominent in both states, resulting in an $A$-always equilibrium, rather than a split one.

\begin{figure}
\begin{minipage}{0.5\textwidth}
\begin{tikzpicture}[scale=0.85,  >=Latex,]

\draw[->, line width=0.5pt] (0,-1) -- (0,5);
\fill (7,-1) node[right] {\footnotesize{$\mu$}};
\draw[->, line width=0.5pt] (0,-1) -- (7,-1);

\draw[line width=1pt,blue] (0,1)--(3,1);
\draw[line width=1pt,blue] (3,1)--(3,4.5);
\draw[line width=1pt,blue] (3,4.5)--(7,4.5);

 \fill (4,-1) node[below] {\footnotesize{$\theta_B$}};
\draw[line width=0.5pt,black,dotted] (3,-1)--(3,2.5);
 \fill (0,4.5) node[left] {\scriptsize{$1+f$}};
 \fill (0,1) node[left] {\footnotesize{$r+f$}};
\fill (0,-1) node[left] {\footnotesize{$0$}};
\draw[line width=0.5pt,black,dotted] (0,4.5)--(7,4.5);
\draw[line width=0.5pt,black,dotted] (0,3.5)--(7,3.5);

 \fill (3,-1) node[below] {\footnotesize{$\theta_A$}};

\fill (0,3.5) node[left] {\footnotesize{$1$}};

\draw[line width=1pt,red] (0,-1)--(4,-1);
\draw[line width=1pt,red] (4,3.5)--(7,3.5);
\draw[line width=1pt,red] (4,-1)--(4,3.5);
\fill (6,3.78) node[right] {\footnotesize{${\overline{z}^1_B(\cdot)}$}};
\fill (6,4.8) node[right] {\footnotesize{$\overline{z}^1_A(\cdot)$}};

\end{tikzpicture}
\end{minipage}
\begin{minipage}{0.5\textwidth}
\begin{tikzpicture}[scale=0.85,  >=Latex,]

\draw[->, line width=0.5pt] (0,-1) -- (0,5);
\fill (7,-1) node[right] {\footnotesize{$\mu$}};
\draw[->, line width=0.5pt] (0,-1) -- (7,-1);

\draw[line width=1pt,blue] (0,0)--(3,0);
\draw[line width=1pt,blue] (3,0)--(3,3.5);
\draw[line width=1pt,blue] (3,3.5)--(7,3.5);
 \fill (4,-1) node[below] {\footnotesize{$\theta_B$}};
\draw[line width=0.5pt,black,dotted] (3,-1)--(3,0.5);
\fill (0,0) node[left] {\footnotesize{$f$}};
\draw[line width=0.5pt,black,dotted] (0,0)--(7,0);
 \fill (0,3.5) node[left] {\scriptsize{$f+1-r$}};
\fill (0,-1) node[left] {\footnotesize{$0$}};
\draw[line width=0.5pt,black,dotted] (0,3.5)--(7,3.5);
\draw[line width=0.5pt,black,dotted] (0,2.5)--(7,2.5);

 \fill (3,-1) node[below] {\footnotesize{$\theta_A$}};

\fill (0,2.5) node[left] {\footnotesize{$1-r$}};

\draw[line width=1pt,red] (0,-1)--(4,-1);
\draw[line width=1pt,red] (4,2.5)--(7,2.5);
\draw[line width=1pt,red] (4,-1)--(4,2.5);
\fill (6,2.78) node[right] {\footnotesize{${\overline{z}_B^0(\cdot)}$}};
\fill (6,3.78) node[right] {\footnotesize{$\overline{z}^0_A(\cdot)$}};

\end{tikzpicture}
\end{minipage}

\caption{\label{fig:equilibrium under opacity} \textit{Proposition \ref{prop:equilibrium under opacity} (Equilibria With Opacity)}. Explanation in text.}
\end{figure}
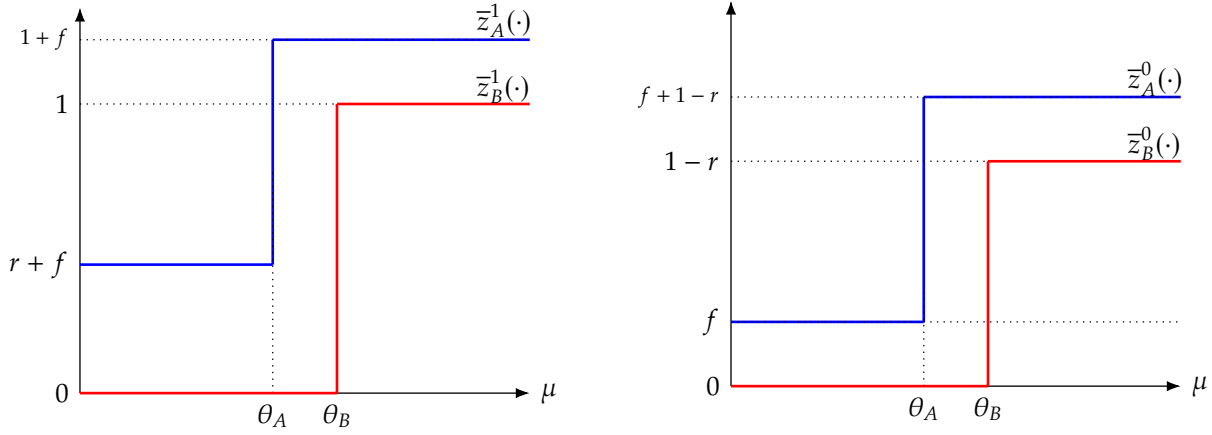

More broadly, Proposition \ref{prop:equilibrium under opacity} has a number of interesting implications. Point (i) implies that an equilibrium in which $A$ is always prominent exists and is most profitable for the platform at each prior belief.  In other words, in its preferred equilibrium, the platform steers the consumer as aggressively as possible by always making $A$ prominent, thereby rendering prominence uninformative about product fit ($\mu_A=\prior$). 

Under opacity, however, the platform cannot directly affect how the consumer interprets prominence. Thus, informative interpretations ($\mu_A\neq \prior$ and $\mu_B\neq\prior$) can also be self-fulfilling. In particular, point (iii) implies that a multiplicity of split equilibria also exists at high priors ($\prior>\theta_A$), and that such equilibria are the only ones in which prominence is informative. Furthermore, point (iv) implies that the platform's profit in the informative split equilibria is strictly smaller than in the uninformative equilibria where $A$ is always prominent. In other words, without direct control over how prominence is perceived, the platform is strictly worse off whenever the consumer interprets prominence as informative.%

This conflict between informativeness and profit arises from an intriguing property: in any informative equilibrium, prominence must convey bad news ($\mu_A<\prior<\mu_B$). Intuitively, the platform is tempted to always make $A$ prominent in order to capitalize on the search friction, which steers the consumer toward it. If prominence also reveals good news, then the temptation to make $A$ prominent becomes stronger, and the platform cannot resist. But always making $A$ prominent eliminates any useful information for the consumer ($\mu_A=\prior$). Thus, informative equilibria are sustained by the consumer's negative interpretation of prominence, which offsets the direct benefit arising from the search friction.

These observations suggest that the platform would benefit from an instrument that selects the $A$-always equilibrium. One possibility is a non-binding social norm that discourages the platform from using customer information when determining product order. Under such a norm, the platform's display algorithm does not depend on the state $\omega$. Anticipating the platform follows the norm, the consumer's belief is unchanged, regardless of which product is prominent ($\mu_A=\mu_B=\prior$). Consequently, it is optimal for the platform to always make $A$ prominent, thereby following the norm. Thus, a social norm against personalized search order uniquely selects the platform's preferred equilibrium.  Similar logic also applies if the platform is allowed to make a non-binding announcement of its display algorithm. If it claims that its algorithm always makes $A$ prominent and it is believed, then such an algorithm is indeed optimal.

\subsection{Transparency}\label{sec:transparency}

We now move to the transparent case, in which the consumer observes the display algorithm when she arrives on the platform. Thus, the consumer's inferences are based on the algorithm that was designed and implemented by the platform. By adjusting the algorithm, the platform changes both the probability that each product is prominent, and the interpretation of prominence itself. This ability to shape the consumer's beliefs introduces a persuasion motive into the platform's design problem, with significant implications.

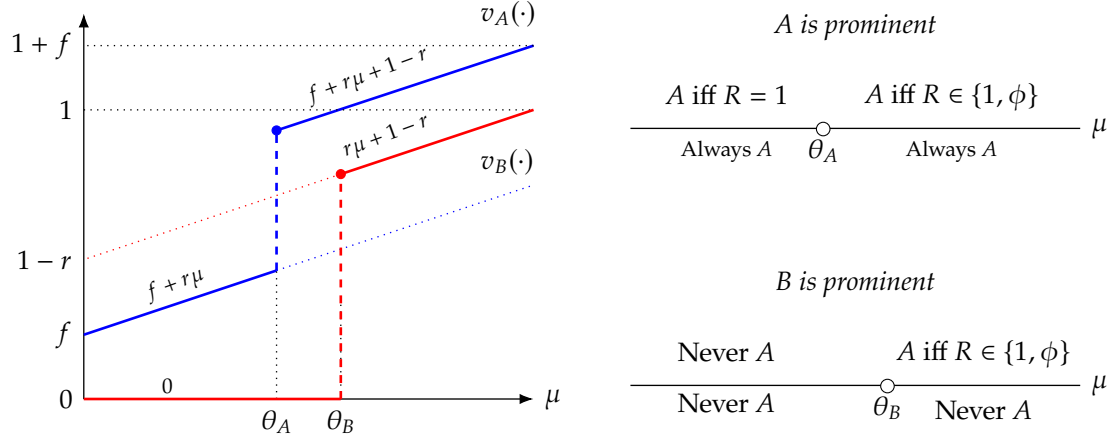
\begin{figure}
\begin{minipage}{0.5\textwidth}
\begin{tikzpicture}[scale=0.85,  >=Latex,]

\draw[->, line width=0.5pt] (0,-1) -- (0,5);
\fill (7,-1) node[right] {\footnotesize{$\mu$}};
\draw[->, line width=0.5pt] (0,-1) -- (7,-1);

\draw[line width=1pt,blue] (0,0)--(3,1);
\draw[line width=1pt,blue, dashed] (3,1)--(3,3.18);
\draw[line width=1pt,blue] (3,3.18)--(7,4.5);
\draw[line width=0.5pt,blue,dotted] (3,1)--(7,2.33);
\draw[line width=0.5pt,black,dotted] (3,-1)--(3,1);
\fill (6,5) node[right] {\footnotesize{$v_A(\cdot)$}};
\fill (4,-1) node[below] {\footnotesize{$\theta_B$}};
\draw[line width=0.5pt,black,dotted] (4,-1)--(4,0.5);
\fill (0,0) node[left] {\footnotesize{$f$}};
\draw[line width=0.5pt,black,dotted] (7,4.5)--(0,4.5);
\fill (0,4.5) node[left] {\footnotesize{$1+f$}};
\fill (0,3.5) node[left] {\footnotesize{$1$}};
\fill (0,-1) node[left] {\footnotesize{$0$}};
\draw[line width=0.5pt,black,dotted] (0,3.5)--(7,3.5);

\fill (3,-1) node[below] {\footnotesize{$\theta_A$}};

\draw[line width=1pt,red] (0,-1)--(4,-1);
\draw[line width=1pt,red] (4,2.5)--(7,3.5);
\draw[line width=1pt,red,dashed] (4,-1)--(4,2.5);
\draw[line width=0.5pt,red,dotted] (4,2.5)--(0,1.17);
\fill (0,1.17) node[left] {\footnotesize{$1-r$}};
\fill (6,2.7) node[right] {\footnotesize{$v_B(\cdot)$}};

\node[below=0pt, rotate=20] at (4.6,3.4) 
{\scriptsize{$r\mu+1-r$}};
\node[below=0pt, rotate=20] at (4.3,4.35) 
{\scriptsize{$f+r\mu+1-r$}};

\node[below=0pt, rotate=20] at (1.3,1.1) 
{\scriptsize{$f+r\mu$}};
\node[below=0pt] at (1.3,-0.5) 
{\scriptsize{$0$}};

\filldraw[color=red, fill=red] (4,2.5) circle (2pt);
\filldraw[color=blue, fill=blue] (3,3.18) circle (2pt);

\end{tikzpicture}
\end{minipage}
\begin{minipage}{0.5\textwidth}
\begin{tikzpicture}[scale=0.85,  >=Latex,]
\pgfmathsetmacro{\off}{4}
\pgfmathsetmacro{\offsm}{1}

\draw[line width=0.5pt] (0,0) -- (7,0);
\fill (7,0) node[right] {\footnotesize{$\mu$}};
\fill (3,0) node[below] {\footnotesize{$\theta_A$}};
\filldraw[color=black, fill=white] (3,-0.01) circle (3pt);
\node[above=0pt] at (3.5,1.2) 
{\footnotesize{$A$ \textit{is prominent}}};

\node[above=0pt] at (5,0.1) 
{\footnotesize{$A$ iff $R\in\{1,\phi\}$}};
\node[above=0pt] at (5,0.1-0.8) 
{\scriptsize{Always $A$}};

\node[above=0pt] at (1.5,0.2) 
{\footnotesize{$A$ iff $R=1$}};
\node[above=0pt] at (1.5,0.1-0.8) 
{\scriptsize{Always $A$}};

\draw[line width=0.5pt] (0,-\off) -- (7,-\off);
\fill (7,-\off) node[right] {\footnotesize{$\mu$}};
\fill (4,-\off) node[below] {\footnotesize{$\theta_B$}};
\filldraw[color=black, fill=white] (4,-0.01-\off) circle (3pt);
\node[above=0pt] at (3.5,1.2-\off) 
{\footnotesize{$B$ \textit{is prominent}}};

\node[above=0pt] at (5.5,0.1-\off) 
{\footnotesize{$A$ iff $R\in\{1,\phi\}$}};
\node[above=0pt] at (5.5,0.1-\off-0.8) 
{\footnotesize{Never $A$}};

\node[above=0pt] at (1.5,0.2-\off) 
{\footnotesize{Never $A$}};
\node[above=0pt] at (1.5,0.2-\off-0.8) 
{\footnotesize{Never $A$}};

\end{tikzpicture}
\end{minipage}

\caption{\label{fig:transparent profit} \textit{Lemma \ref{lem:transparent profit} (Transparent Profit)}. Left: The payoffs $v_i(\cdot)$, where $i\in\{A,B\}$ is the prominent product. Under regularity (ii) the dotted red line lies above the blue at 0. Right: consumer search strategy. Above line, \searcher{}; below, \captive{}. Computing the probability $A$ is bought using common belief $\mu$ and normalizing recovers $v_i(\cdot)$.}
\end{figure}

\paragraph{Transparent Profit.} When the algorithm is transparent, the posterior beliefs $(\mu_A,\mu_B)$ are determined by the algorithm itself $\Pi=(\pi_0,\pi_1)$, not by a conjecture. Thus, the \searcher{}'s continuation strategy depends on the algorithm through these induced posteriors. This allows for two simplifications to the platform's expected profit (see  \eqref{eq:opaque profit} for the opaque case). First, the links between the algorithm and the induced posteriors (described in Section \ref{sec:belief-based approach}) allow the platform's profit to be written as a function of the posterior beliefs $(\mu_A,\mu_B)$, which it chooses directly.\footnote{In the opaque case, this reduction is possible in equilibrium, which imposes that the conjectured algorithm $\widehat{\Pi}$ and the true algorithm $\Pi$ are the same. When designing its algorithm, however, the platform is free to depart from the conjecture, so the reductions need not hold.} Second, the \searcher{} must break ties in the platform's favor in equilibrium; otherwise, the platform profits from a marginal deviation.

\begin{lemma}[\textbf{Transparent Profit}]\label{lem:transparent profit} In equilibrium under transparency, the platform maximizes
\[\mathbb{E}_\tau[v_i(\mu_i)]=\tau_Av_A(\mu_A)+\tau_Bv_B(\mu_B),\]
 where \(v_A(\mu)=f+r\mu+(1-r)\mathbbm{1}\{\mu\geq\theta_A\}\), and \(v_B(\mu)=(r\mu+1-r)\mathbbm{1}\{\mu\geq\theta_B\}\).
\end{lemma}
  Note that $\mathbbm{1}\{\cdot\}$ is the usual indicator function. With favorable tie-breaking, $v_i(\mu)$ is the platform's expected payoff if product $i\in\{A,B\}$ is  prominent and the platform and consumer have a common belief $\mu$ (see Figure \ref{fig:transparent profit}). To understand this, imagine that the algorithm is able to observe the state, but the algorithm designer, i.e., the platform, is not. This is payoff-equivalent for the platform, as the platform has no information about the state at the design stage, and its subsequent actions are determined by the algorithm. Thus, when belief $\mu_i$ is realized by the algorithm, it is common to the platform and consumer, and the platform's expected payoff is $v_i(\mu_i)$. The platform's payoff at the design stage is then computed by the law of iterated expectations, $\mathbb{E}_\tau[v_i(\mu_i)]$.
  
    The difference $v_A(\mu)-v_B(\mu)>0$ reflects the steering effect of prominence, the increased sales arising solely from $A$'s advantageous position in the search order. Meanwhile, the spread in the posterior beliefs $(\mu_A,\mu_B)$ reflects the information effect. By choosing these beliefs, the platform determines what information is conveyed, persuading the consumer.

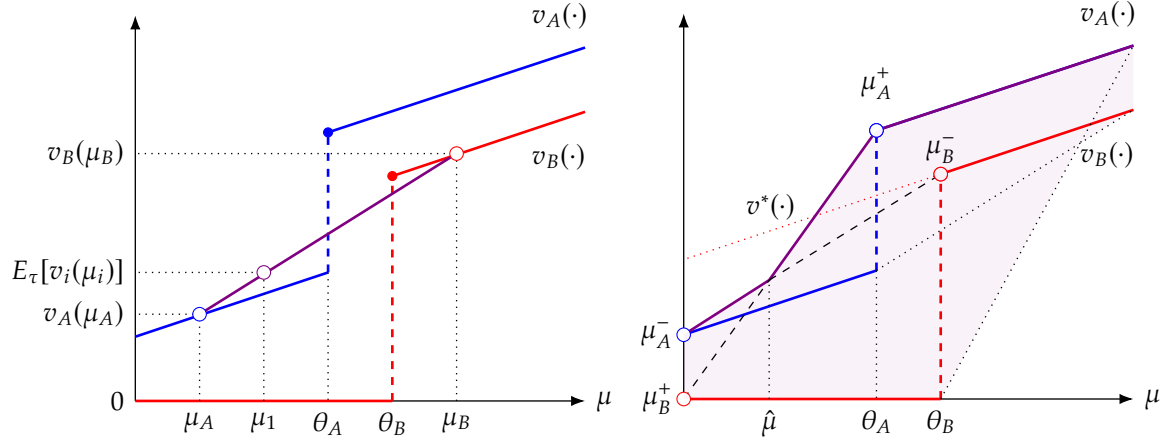
\begin{figure}
\begin{minipage}{0.5\textwidth}
\begin{tikzpicture}[scale=0.85,  >=Latex]

\draw[->, line width=0.5pt] (0,-1) -- (0,5);
\fill (7,-1) node[right] {\footnotesize{$\mu$}};
\draw[->, line width=0.5pt] (0,-1) -- (7,-1);

\draw[line width=1pt,blue] (0,0)--(3,1);
\draw[line width=1pt,blue, dashed] (3,1)--(3,3.18);
\draw[line width=0.5pt, dotted] (3,1)--(3,-1);
\draw[line width=1pt,blue] (3,3.18)--(7,4.5);
\fill (6,5) node[right] {\footnotesize{$v_A(\cdot)$}};
\fill (4,-1) node[below] {\footnotesize{$\theta_B$}};
\fill (0,-1) node[left] {\footnotesize{$0$}};

\fill (3,-1) node[below] {\footnotesize{$\theta_A$}};

\draw[line width=1pt,red] (0,-1)--(4,-1);
\draw[line width=1pt,red] (4,2.5)--(7,3.5);
\draw[line width=1pt,red, dashed] (4,-1)--(4,2.5);
\fill (6,2.8) node[right] {\footnotesize{$v_B(\cdot)$}};

\filldraw[color=red, fill=red] (4,2.5) circle (2pt);
\filldraw[color=blue, fill=blue] (3,3.18) circle (2pt);

\pgfmathsetmacro{\zA}{0.35}
\pgfmathsetmacro{\zB}{2.85}
\pgfmathsetmacro{\muA}{1}
\pgfmathsetmacro{\muB}{5}
\pgfmathsetmacro{\pri}{2}
\pgfmathsetmacro{\zpri}{1}

\draw[line width=1pt,violet] (\muA,\zA)--(\muB,\zB);

\draw[line width=0.5pt,black,dotted] (\muB,-1)--(\muB,\zB)--(0,\zB);

\fill (0,\zB) node[left] {\footnotesize{$v_B(\mu_B)$}};
\fill (\muB,-1) node[below] {\footnotesize{$\mu_B$}};
\filldraw[color=red, fill=white] (\muB,\zB) circle (3pt);

\draw[line width=0.5pt,black,dotted] (\muA,-1)--(\muA,\zA)--(0,\zA);
\fill (0,\zA) node[left] {\footnotesize{$v_A(\mu_A)$}};
\fill (\muA,-1) node[below] {\footnotesize{$\mu_A$}};
\filldraw[color=blue, fill=white] (\muA,\zA) circle (3pt);

\draw[line width=0.5pt,black,dotted] (\pri,-1)--(\pri,\zpri)--(0,\zpri);
\fill (0,\zpri) node[left] {\footnotesize{$E_\tau[v_i(\mu_i)]$}};
\fill (\pri,-1) node[below] {\footnotesize{$\prior$}};
\filldraw[color=violet, fill=white] (\pri,\zpri) circle (3pt);

\end{tikzpicture}
\end{minipage}
\begin{minipage}{0.5\textwidth}
\begin{tikzpicture}[scale=0.85,  >=Latex,]

\draw[->, line width=0.5pt] (0,-1) -- (0,5);
\fill (7,-1) node[right] {\footnotesize{$\mu$}};
\draw[->, line width=0.5pt] (0,-1) -- (7,-1);

\draw[line width=1pt,blue] (0,0)--(3,1);
\draw[line width=1pt,blue, dashed] (3,1)--(3,3.18);
\draw[line width=0.5pt, dotted] (3,1)--(3,-1);
\draw[line width=1pt,blue] (3,3.18)--(7,4.5);
\fill (6,5) node[right] {\footnotesize{$v_A(\cdot)$}};
\fill (0.8,2) node[right] {\footnotesize{$v^*(\cdot)$}};
\fill (4,-1) node[below] {\footnotesize{$\theta_B$}};

\fill (3,-1) node[below] {\footnotesize{$\theta_A$}};

\draw[line width=1pt,red] (0,-1)--(4,-1);

\draw[line width=1pt,red] (4,2.5)--(7,3.5);
\draw[line width=1pt,red, dashed] (4,-1)--(4,2.5);
 \draw[line width=0.5pt,red,dotted] (4,2.5)--(0,1.17);
\fill (6,2.8) node[right] {\footnotesize{$v_B(\cdot)$}};

\draw[line width=0.5pt,dotted] (4,-1)--(7,4.5);
\draw[line width=0.5pt,dotted] (3,1)--(7,3.5);

\filldraw[color=red, fill=red] (4,2.5) circle (3pt);
\filldraw[color=blue, fill=blue] (3,3.18) circle (3pt);

\draw[line width=1pt,violet] (0,0)--(1.33,0.85)--(3,3.18);
\draw[line width=1pt,violet] (3,3.18)--(7,4.5);
\draw[line width=0.5pt, dashed] (1.33,0.85) -- (4,2.5);
\draw[line width=0.5pt, dashed] (0,-1) -- (1.33,0.85);
\fill (0,0) node[left] {\footnotesize{$\mu_{A}^{-}$}};
\fill (3,3.5) node[above] {\footnotesize{$\mu_{A}^{+}$}};
\fill (4,2.5) node[above] {\footnotesize{$\mu_{B}^{-}$}};
\fill (0,-1) node[left] {\footnotesize{$\mu_{B}^{+}$}};
\filldraw[color=red, fill=white] (4,2.5) circle (3pt);
\filldraw[color=red, fill=white] (0,-1) circle (3pt);
\filldraw[color=blue, fill=white] (3,3.18) circle (3pt);
\filldraw[color=blue, fill=white] (0,0) circle (3pt);

\fill[violet, opacity=0.05] (0,0) -- (1.33,0.85) -- (3,3.18) -- (7,4.5) --(7,3.5) --({179/29},{173/58})--(4,-1)--(0,-1) -- cycle;

\draw[line width=0.5pt, dotted] (1.33,-1) -- (1.33,0.85); 
\fill (1.33,-1) node[below] {\footnotesize{$\hat{\mu}$}};

\end{tikzpicture}
\end{minipage}
\caption{\label{fig:equilibrium under transparency} \textit{Proposition \ref{prop:equilibrium under transparency}}. Left panel: the attainable profits with posterior beliefs $(\mu_A,\mu_B)$ for all possible prior beliefs are represented by the purple line segment. Right panel: the set of all attainable profits is shaded purple. Its upper envelope $v^*(\cdot)$ is the highest attainable profit at each prior. For priors $\prior\in[\widehat{\mu},\theta_A]$, inducing beliefs $(\mu_A^+,\mu_B^+)$ via positive prominence generates highest profit. For $\prior\in[0,\widehat{\mu}]$, inducing beliefs $(\mu_A^-,\mu_B^-)$ via negative prominence generates highest profit.}
\end{figure}

\paragraph{Transparent Algorithm Design.} As described in Section \ref{sec:belief-based approach}, an algorithm $\Pi$ can be represented by a pair of beliefs $(\mu_A,\mu_B)$ surrounding the prior $\prior$, with probabilities $(\tau_A,\tau_B)$ that satisfy the law of iterated expectations. The expected profit of such an algorithm is evaluated under the same distribution, so that
\[\mathbb{E}_\tau[\mu_i]=\tau_A\mu_A+\tau_B\mu_B=\prior\quad\quad \mathbb{E}_\tau[v_i(\mu_i)]=\tau_Av_A(\mu_A)+\tau_Bv_B(\mu_B).\]
These observations allow for a straightforward geometric characterization, illustrated in Figure \ref{fig:equilibrium under transparency}. Graphically, the expected profit of an algorithm $(\mu_A,\mu_B)$ can be found by drawing a line segment joining points $(\mu_A,v_A(\mu_A))$ and $(\mu_B,v_B(\mu_B))$, and then evaluating its height at prior $\prior$, as illustrated in the left panel. More broadly, the closed line segment represents all profit levels generated by induced posteriors $(\mu_A,\mu_B)$ as the prior varies over the unit interval. Varying over all possible posteriors yields the highest attainable profit at each prior belief, $v^*(\cdot)$, drawn in purple in the right panel. At each prior belief, the beliefs that are joined to generate this curve are induced by the equilibrium algorithm.

\paragraph{Equilibrium Under Transparency.} In anticipation of the formal result to follow, we define two algorithms that play an important role in the equilibrium characterization.

\begin{definition}[\textbf{Positive and Negative Sorting}] For prior $\prior\in(0,\theta_A)$, 
\begin{itemize}
\item the \textbf{positive sorting} algorithm induces posteriors $(\mu_A=\theta_A,\,\mu_B= 0)$.
\item the \textbf{negative sorting} algorithm induces posterior beliefs $(\mu_A=0,\,\mu_B=\theta_B)$.
\end{itemize}
\end{definition}
 We sometimes denote the posteriors under positive sorting as $(\mu_A^+,\mu_B^+)\equiv(\theta_A,0)$, and negative sorting as $(\mu_A^-,\mu_B^-)\equiv(0,\theta_B)$. 

 In the positive (negative) sorting algorithm prominence conveys good (bad) news. 
 The terminology reflects how prominence is assigned. Under positive sorting, each product is more likely to be prominent when it is the better match than when it is worse ($\pi_1>\pi_0$), and under
negative sorting, the reverse ($\pi_1<\pi_0$). If the platform would like to persuade the consumer to change her search decision, it is optimal to use one of these algorithms.\footnote{Consult the left panel of Figure \ref{fig:equilibrium under transparency}. In order to change the consumer's search decision, the platform must induce either $\mu_B\geq\theta_B$ or $\mu_A\geq\theta_A$.  Suppose $\mu_B\geq\theta_B$, which implies $\mu_A<\prior$. The highest profit is attained by shifting $\{\mu_A,\mu_B\}$ as far left as possible, $\{\mu_A=0,\mu_B=\theta_B\}$. The argument for $\mu_A\geq\theta_A$ is analogous.}

Positive and negative sorting persuade the \searcher{} in opposite directions. Under positive sorting, $A$'s prominence conveys just enough good news to deter continued search when the signal is inconclusive ($R=\phi$). Meanwhile, $B$'s prominence reveals it to be the best match, and it is therefore bought immediately. Thus, with positive sorting, the \searcher{} continues only if $A$ is prominent and she learns that $B$ is her best match ($R=0$), the only scenario in which continued search cannot possibly be prevented. In other words, positive sorting maximally deters continued search. In contrast, negative sorting maximally encourages it. When $A$ is prominent, it is revealed to be the worse match; when $B$ is prominent, the belief about $A$ improves just enough that the \searcher{} is willing to incur the cost to continue. Thus, with negative sorting, the \searcher{} always encounters both products, regardless of which product is prominent.

The following proposition follows from inspection of the right panel of Figure \ref{fig:equilibrium under transparency}.

\begin{proposition}[\textbf{Equilibrium Under Transparency}]\label{prop:equilibrium under transparency} When the display algorithm is transparent, a unique $\widehat{\mu}\in(0,\theta_A)$ exists such that
\begin{enumerate}[label=\textup{(\roman*)}]
\item if $0<\prior<\widehat{\mu}$, then the unique equilibrium algorithm is negative sorting.
\item if $\widehat{\mu}<\prior<\theta_A$, then the unique equilibrium algorithm is positive sorting.
\item if $\prior\geq \theta_A$ or $\prior=0$, then the unique equilibrium algorithm always makes $A$ prominent.
\end{enumerate}

\end{proposition}

To understand Proposition \ref{prop:equilibrium under transparency}, first consider high priors ($\prior\geq\theta_A$). By always making $A$ prominent, the platform ensures that the \captive{} always buy it, and that the \searcher{} buys it unless she learns conclusively that $B$ is her best match. The platform could not possibly sell $A$ with higher probability.\footnote{In any algorithm, the \searcher{} buys $B$ if she conclusively learns that it is her best match. As described in text, when $A$ is prominent and $\prior\geq\theta_A$, the \searcher{} buys $A$ in all other scenarios.} Thus, the platform's highest profit is attained by always making $A$ prominent, as in its preferred equilibrium under opacity. In this sense, at high prior beliefs the platform's only benefit of transparency is equilibrium selection.

In contrast, when the prior is less favorable ($\prior<\theta_A$), an algorithm that always makes $A$ prominent is uninformative, and the \searcher{} buys $A$ only when the signal reveals conclusively that it is a match ($R=1$). The platform can do better with many algorithms that convey information and persuade the \searcher{} to change her search strategy. As an example, consider the algorithm illustrated in the left panel of Figure \ref{fig:equilibrium under transparency}, where $\mu_A<\prior<\theta_B<\mu_B$. When $A$ is prominent, the \searcher{} becomes more pessimistic and follows the same search strategy as under the prior. When $B$ is prominent, however, the \searcher{} becomes sufficiently optimistic about $A$ to continue, which increases the probability that she ultimately buys $A$. By implication, the equilibrium algorithm must convey enough information to change the consumer's search decision. Recalling that positive and negative sorting are optimal among algorithms with this property, it follows that the equilibrium algorithm must be one of these two.

The tradeoff between positive and negative sorting hinges on how these algorithms align the choices of \captive{} and \searcher{}. With positive sorting, neither type of consumer buys \(A\) when \(B\) is prominent, but when $A$ is prominent, both types of consumer buy $A$ as often as possible: the \captive{} always buys $A$, and the \searcher{} buys $A$ unless she conclusively learns that it is a mismatch ($R=0$). Thus, positive sorting aligns the decisions of both types, generating a high profit when \(A\) is prominent and zero profit otherwise. In contrast, negative sorting misaligns the types' choices. Whenever $A$ is prominent, the \captive{} always buys it, but the \searcher{} never does. Whenever $B$ is prominent, the \captive{} always buys it, but the \searcher{} instead buys $A$ as often as possible. Because of this misalignment, negative sorting generates a moderate profit regardless of which product is actually prominent. By implication, positive sorting is preferred at moderate priors, where it is relatively likely to make $A$ prominent and generate a high profit.\footnote{Under positive sorting, $\tau_A=\prior/\theta_A$.} At low priors, positive sorting is too likely to make $B$ prominent, generating zero profit, and negative sorting is preferred.

\begin{remark}[\textbf{Consumer Heterogeneity}]
If the consumer is a known \captive{}, then the unique equilibrium algorithm always makes $A$ prominent. If the consumer is a known \searcher{}, then the unique equilibrium algorithm is positive sorting.\footnote{This is evident from the right panel of Figure \ref{fig:equilibrium under transparency}. Imagine that the consumer is known to be a \searcher{}, so that $f=0$. In this case, the optimal algorithm always interpolates $\mu_B=0$ on the red graph and $\mu_A=\theta_A$ on the blue, i.e., the positive algorithm.} Thus, negative sorting is used in equilibrium due to the interplay of incentives between the \captive{} and the \searcher{}. 
\end{remark}

\begin{remark}[\textbf{Maximizing Sales of $B$}]\label{rem:maxB} Suppose that the commission is higher for product $B$, so that the platform would like to maximize the probability $B$ is purchased. Because the consumer always buys either $A$ or $B$, this is equivalent to minimizing the probability $A$ is purchased. The construction is analogous to Proposition \ref{prop:equilibrium under transparency}, except that the objective is reversed. Up to normalization, the lowest possible probability that $A$ sells is the lower bound of the shaded purple region in the right panel of Figure \ref{prop:equilibrium under transparency}. As is evident, when minimizing the probability that $A$ sells (which maximizes $B$), the platform uses versions of the positive and negative sorting algorithms when $\prior>\theta_B$.

\end{remark}

\paragraph{Search Friction.} We briefly consider the effect of the search friction, reflected in the parameter $c$.\footnote{The other component of the search friction has a straightforward effect. As $h$ increases, $f$ increases, shifting the blue graph up. Thus, both the positive and negative algorithms become more profitable.} As discussed previously,  positive and negative sorting incentivize opposite search behaviors. When positive sorting makes $A$ prominent, the \searcher{}'s posterior belief is $\mu_A=\theta_A<\out$, and the searcher with an inconclusive realization stops and buys $A$, even though she would prefer to buy $B$ in the absence of search cost. Thus, positive sorting exploits the search cost to deter search when $A$ is prominent. With negative sorting, the \searcher{} never buys $A$ when it is prominent. However, when $B$ is prominent, the \searcher{}'s posterior belief is $\mu_B=\theta_B>\out$, and she is incentivized to continue and encounter $A$. Thus, negative sorting must overcome the search cost in order to sell $A$ when $B$ is prominent. 

These observations motivate the following corollary, which follows from the inspection of Figure \ref{fig:search friction}.

\begin{figure}
\begin{minipage}{0.5\textwidth}
\begin{tikzpicture}[scale=0.85,  >=Latex]

\draw[->,line width=0.5pt] (0,-1) -- (0,5);
\fill (7,-1) node[right] {\footnotesize{$\mu$}};
\draw[->,line width=0.5pt] (0,-1) -- (7,-1);

\draw[line width=0.5pt, dashed] (1.33,0.85) -- (4,2.5);
\draw[line width=0.5pt, dashed] (0,-1) -- (1.33,0.85);

\draw[line width=1pt,violet] (0,0)--(1.33,0.85)--(3,3.18);

\draw[line width=1pt,blue] (0,0)--(3,1);
\draw[line width=0.5pt,blue,dashed] (3,1)--(3,3.18);
\draw[line width=1pt,violet] (3,3.18)--(7,4.5);
\draw[line width=0.5pt,black,dotted] (3,-1)--(3,1);
\fill (6,5) node[right] {\footnotesize{${v}_A(\cdot)$}};
\fill (4,-1) node[below] {\footnotesize{$\theta_B$}};
\fill (0,0) node[left] {\footnotesize{$\mu_{A}^{-}$}};
\fill (3,3.5) node[above] {\footnotesize{$\mu_{A}^{+}$}};
\fill (4,2.5) node[above] {\footnotesize{$\mu_{B}^{-}$}};

\draw[line width=0.5pt,black,dotted] (7,4.5)--(0,4.5);
\fill (0,4.5) node[left] {\footnotesize{$1+f$}};
\fill (0,3.5) node[left] {\footnotesize{$1$}};
\fill (0,-1) node[left] {\footnotesize{$\mu_{B}^{+}$}};
\draw[line width=0.5pt,black,dotted] (0,3.5)--(7,3.5);

\fill (3,-1) node[below] {\footnotesize{$\theta_A$}};

\draw[line width=1pt,red] (0,-1)--(4,-1);
\draw[line width=1pt,red] (4,2.5)--(7,3.5);
\draw[line width=0.5pt,red,dashed] (4,-1)--(4,2.5);
\fill (6,2.7) node[right] {\footnotesize{${v}_B(\cdot)$}};

\draw[->,line width=0.5pt] (3-0.4,-1.3) -- (2.3-0.4,-1.3);
\draw[->,line width=0.5pt] (4.4,-1.3)--(5.1,-1.3);

\filldraw[color=red, fill=white] (4,2.5) circle (3pt);
\filldraw[color=red, fill=white] (0,-1) circle (3pt);
\filldraw[color=blue, fill=white] (3,3.18) circle (3pt);
\filldraw[color=blue, fill=white] (0,0) circle (3pt);

\draw[line width=0.5pt,blue,dotted] (3,1)--(7,2.33);
\draw[line width=0.5pt,red,dotted] (4,2.5)--(0,1.17);
\draw[line width=0.5pt,blue,dotted] (4,3.5)--(0,2.17);
\end{tikzpicture}
\end{minipage}
\begin{minipage}{0.5\textwidth}
\begin{tikzpicture}[scale=0.85,  >=Latex]

\draw[->,line width=0.5pt] (0,-1) -- (0,5);
\fill (7,-1) node[right] {\footnotesize{$\mu$}};
\draw[->,line width=0.5pt] (0,-1) -- (7,-1);

\draw[line width=0.5pt, dashed] (0,-1) -- (2,2.85);
\draw[line width=0.5pt, dashed] (0,0) -- (4.5,2.66);

\draw[line width=1pt,violet] (0,0)--(0.73,0.44)--(2,2.85);

\draw[line width=1pt,blue] (0,0)--(2,0.66);
\draw[line width=0.5pt,blue,dashed] (2,0.66)--(2,2.85);
\draw[line width=1pt,violet] (2,2.85)--(7,4.5);
\draw[line width=0.5pt,black,dotted] (2,-1)--(2,0.66);
\fill (6,5) node[right] {\footnotesize{${v}_A(\cdot)$}};
\fill (4.5,-1) node[below] {\footnotesize{$\theta_B$}};
\fill (0,0) node[left] {\footnotesize{$\mu_{A}^{-}$}};
\fill (1.7,2.8) node[above] {\footnotesize{$\mu_{A}^{+}$}};
\fill (4,2.5) node[above] {\footnotesize{$\mu_{B}^{-}$}};

\draw[line width=0.5pt,black,dotted] (7,4.5)--(0,4.5);
\fill (0,4.5) node[left] {\footnotesize{$1+f$}};
\fill (0,3.5) node[left] {\footnotesize{$1$}};
\fill (0,-1) node[left] {\footnotesize{$\mu_{B}^{+}$}};
\draw[line width=0.5pt,black,dotted] (0,3.5)--(7,3.5);

\fill (2,-1) node[below] {\footnotesize{$\theta_A$}};

\draw[line width=1pt,red] (0,-1)--(4.5,-1);
\draw[line width=1pt,red] (4.5,2.66)--(7,3.5);
\draw[line width=0.5pt,red,dashed] (4.5,-1)--(4.5,2.66);
\fill (6,2.7) node[right] {\footnotesize{$v_B(\cdot)$}};

\draw[line width=0.5pt,blue,dotted] (0,0)--(7,2.33);
\draw[line width=0.5pt,red,dotted] (4.5,2.66)--(0,1.17);
\draw[line width=0.5pt,blue,dotted] (4,3.5)--(0,2.17);

\filldraw[color=red, fill=white] (4.5,2.66) circle (3pt);
\filldraw[color=red, fill=white] (0,-1) circle (3pt);
\filldraw[color=blue, fill=white] (2,2.85) circle (3pt);

\filldraw[color=blue, fill=white] (0,0) circle (3pt);

\end{tikzpicture}
\end{minipage}
\caption{\label{fig:search friction}\textit{Corollary \ref{cor:search friction}}. An increase in $c$ shifts $\theta_A$ left and $\theta_B$ right. The profit of the positive algorithm increases, while the profit of the  negative algorithm decreases.}
\end{figure}
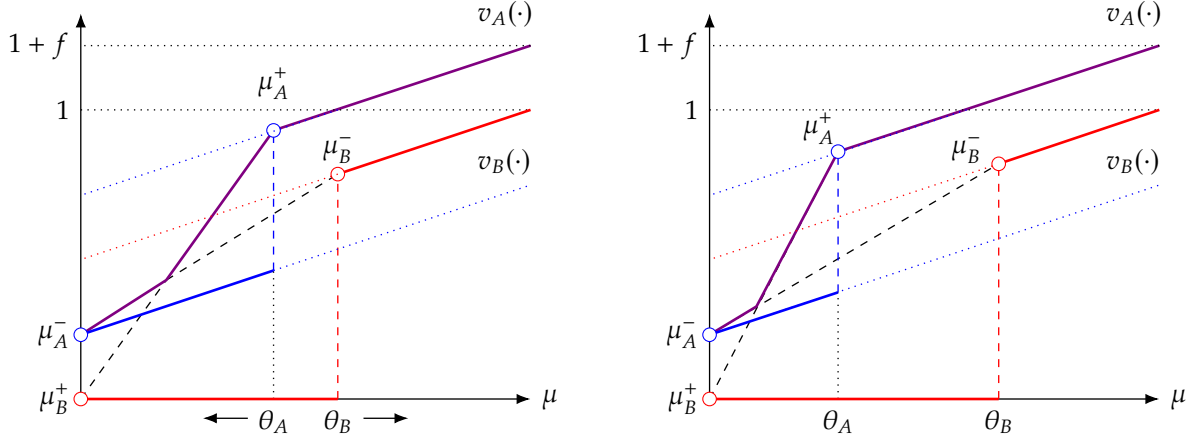

\begin{corollary}[\textbf{Search Friction}]\label{cor:search friction}
When the search cost $c$ increases, the platform's profit from positive sorting increases and the expected profit from negative sorting decreases.
\end{corollary}
As illustrated in Figure \ref{fig:search friction}, an increase in $c$ shifts $\theta_A$ left, and shifts $\theta_B$ right. Thus the line representing the profit of positive algorithm rotates counterclockwise around (0,0), while the profit of negative sorting rotates clockwise around $(0,f)$.

\paragraph{Transparency and Welfare.} We analyze the implications of transparency for welfare. While it may seem plausible, it should not be taken for granted that transparency enhances consumer welfare: under transparency, the consumer understands how prominence arises and what it means, which allows the platform to shape the consumer's inferences. The platform benefits, but the implications for the consumer are not immediately apparent.

Our approach to finding the consumer's equilibrium payoff is similar to our construction of the platform's optimal algorithm and profit. Let $\mu=\Pr(\omega=1)$ be the consumer's belief at the end of stage (ii), and let $i\in\{A,B\}$ be the prominent product. Given her high search cost, it is optimal for the \captive{} to purchase the prominent product. Therefore given belief $\mu$ and prominent product $i$, the \captive{}'s expected payoffs are 
\[u^{cap}_A(\mu)=\mu\quad\quad u^{cap}_B(\mu)=\out.\]
The \searcher{}'s equilibrium strategy is derived in Lemma \ref{lem:searcher strategy} by comparing the expected payoff of continuing and stopping at each combination of product prominence and the stage (iii) belief. The \searcher{}'s stage (iii) payoff under optimal search is the larger of these, i.e., the upper envelope of the payoffs illustrated in Figure \ref{fig:searcher strategy},
\begin{align*}w_A(\mu)=\max\{\out-c,\mu\}\quad\quad w_B(\mu)=\max\{\out, r(\mu+(1-\mu)\out)+(1-r)\mu-c\}.\end{align*}
Now consider the \searcher{}'s payoff in stage (ii), when product $i$ is prominent and the stage (ii) posterior belief is $\mu$. If $B$ is prominent, then no learning occurs in stage (iii), and the \searcher{}'s payoff is identical in stages (ii) and (iii). If $A$ is prominent, however, the signal realization is observed in stage (iii). To compute the \searcher{}'s expected payoff as a function of the belief at the end of stage (ii) when $A$ is prominent, we must account for the possibility of learning. Thus, the \searcher{}'s stage (ii) payoff functions $u_i(\cdot)$ are, 
 \begin{align*}u^{se}_A(\mu)=r(\mu+(1-\mu)(\out-c))+(1-r)w_A(\mu)\quad\quad u^{se}_B(\mu)=w_B(\mu)
 \end{align*}

To find the equilibrium payoff of consumer type $j\in\{cap,se\}$, we exploit the same observation as for the platform: for any algorithm that induces beliefs $\{\mu_A,\mu_B\}$, 
\[(\prior,\mathbb{E}_\tau[u_i^{j}(\mu_i)])=\tau_A(\mu_A,u_A^{j}(\mu_A))+\tau_B(\mu_B,u_B^{j}(\mu_B)).\]
Thus, the type $j$ consumer's profit from an algorithm can be found in a similar manner: join points $(\mu_A,u_A^j(\mu_A))$ and $(\mu_B,u_B^j(\mu_B))$ with a line segment, and evaluate it at $\prior$. The construction of each type's equilibrium profit is illustrated in Figure \ref{fig:transparency and welfare} and described in detail in the caption. In the left panel, the purple curve is the equilibrium payoff of   the \captive{}; in the right panel, the \searcher{}.

\begin{figure}
\centering

\begin{minipage}{0.48\textwidth}
\centering
\begin{tikzpicture}[
    scale=8.6,
    >=Latex,
    every node/.style={font=\footnotesize}
]

\pgfmathsetmacro{\u}{0.5}
\pgfmathsetmacro{\r}{0.6}
\pgfmathsetmacro{\c}{0.2}

\def\thetaA{\u-\c} %
\pgfmathsetmacro{\thetaB}{(\u*(1-\r)+\c)/(1-\r*\u)}
\pgfmathsetmacro{\off}{0.1}

\draw[->, line width=0.5pt] (0,0) -- (\thetaB+\off,0)
    node[below right] {{$\mu$}};

\draw[->, line width=0.5pt] (0,0) -- (0,\thetaB+\off);

\draw[blue, line width=1pt] (0,0) -- (\thetaB+\off,\thetaB+\off);
\draw[red, line width=1pt] (0,\u) -- (\thetaB+\off,\u);

 \draw[dotted, line width=0.5pt]
       (\u,\u) --(\u,0);
     \draw[dotted, line width=0.5pt]
    (\thetaA,\thetaA) --(\thetaA,0);
\draw[dotted, line width=0.5pt]
    (\thetaB,\u) --(\thetaB,0);

    \draw[violet, line width=1pt] (0,0) -- (0.2,0.175)--(0.2,0.37)--(\thetaA,\thetaA);

 \node[left=4pt] at (\thetaB+\off,\thetaB+\off) {{$u^{cap}_A(\cdot)$}};
 \node[left=4pt] at (0.4,0.55) {{$u^{cap}_B(\cdot)$}};

  \node[below=4pt] at (\u,0) {{$\out$}};
 \node[left=4pt] at (0,\u) {{$\out$}};
\node[below=4pt] at (\thetaA,0) {{$\theta_A$}};
\node[below=4pt] at (\thetaB,0) {{$\theta_B$}};

\draw[line width=0.5pt, dashed] (0,\u) --(\thetaA,\thetaA);

\draw[line width=0.5pt, dashed] (0,0) --(\thetaB,\u);

\filldraw[color=red, fill=white] (0,\u) circle (0.5pt);
\node[right=4pt] at (0,\u+0.05)  {{$\mu_B^+$}};
\filldraw[color=blue, fill=white] (\thetaA,\thetaA) circle (0.5pt);
\node[above=4pt] at (\thetaA,\thetaA) {{$\mu_A^+$}};

\filldraw[color=blue, fill=white] (0,0) circle (0.5pt);
\node[left=4pt] at (0,0) {{$\mu_A^-$}};
\filldraw[color=red, fill=white] (\thetaB,\u) circle (0.5pt);
\node[right=4pt] at (\thetaB-0.03,\u-0.05) {{$\mu_B^-$}};

\draw[dotted, line width=0.5pt] (0.2,0)--(0.2,0.155);
\node[below=4pt] at (0.2,0.01) {{$\widehat{\mu}$}};

\end{tikzpicture}
\end{minipage}
\hfill
\begin{minipage}{0.48\textwidth}
\centering
\begin{tikzpicture}[
    scale=8.6,
    >=Latex,
    every node/.style={font=\footnotesize}
]

\pgfmathsetmacro{\u}{0.5}
\pgfmathsetmacro{\r}{0.6}
 \pgfmathsetmacro{\c}{0.2}
 \pgfmathsetmacro{\off}{0.1}

\pgfmathsetmacro{\uc}{\u-\c}
\pgfmathsetmacro{\thetaA}{\u-\c}
 \pgfmathsetmacro{\x}{\r*(\thetaA+(1-\thetaA)*\u)+(1-\r)*\thetaA-\c}
\pgfmathsetmacro{\ymid}{\u + \r*\u*(1-\u) - \c}
\pgfmathsetmacro{\thetaB}{(\u*(1-\r)+\c)/(1-\r*\u)}
\pgfmathsetmacro{\UAkink}{\thetaA + \r*\thetaA*(1-\thetaA)}
\pgfmathsetmacro{\yr}{\r*(\thetaB+\off+(1-\thetaB-\off)*\u)+(1-\r)*(\thetaB+\off)-\c}

\draw[line width=1pt, violet] (0,\thetaA) --(0.2,0.37)--(0.2,0.45)--(\thetaA,\UAkink);

\node[left=4pt] at (\thetaB+\off,\thetaB+\off) {{$u^{se}_A(\cdot)$}};
\node[left=4pt] at (0.4,0.55) {{$u^{se}_B(\cdot)$}};

\draw[->, line width=0.5pt] (0,0) -- (\thetaB+\off,0)
    node[below right] {{$\mu$}};

\draw[->, line width=0.5pt] (0,0) -- (0,\thetaB+\off);

\draw[red, line width=1pt] (0,\u) -- (\thetaB,\u);

\draw[red, line width=1pt]
    (\thetaB,\u) -- (\thetaB+\off,\yr);

 \draw[dotted, line width=0.5pt]
    (\thetaB,0) -- (\thetaB,\u);

\node[left=4pt] at (0,\u) {{$\out$}};
\node[left=4pt] at (0,\uc) {{$\out-c$}};

\node[below=4pt] at (\thetaB,0) {{$\theta_B$}};

\draw[blue, line width=1pt] (0,\thetaA) -- (\thetaA,\UAkink);
\draw[blue, line width=1pt] (\thetaA,\UAkink) -- (\thetaB+\off,\thetaB+\off);
\draw[dotted, line width=0.5pt]
    (\thetaA,0) -- (\thetaA,\UAkink);

\node[left=4pt] at (0,\thetaA) {\footnotesize{$\out-c$}};
\node[below=4pt] at (\thetaA,0) {{$\theta_A$}};
\node[left=4pt] at (0,\u) {{$\out$}};

\draw[dotted, line width=0.5pt] (0.2,0)--(0.2,0.45);
\node[below=4pt] at (0.2,0+0.01) {{$\widehat{\mu}$}};

\draw[line width=0.5pt, dashed] (\thetaB,\u) --(0,\thetaA);

\draw[line width=0.5pt, dashed] (0,\u) --(\thetaA,\UAkink);

\filldraw[color=red, fill=white] (0,\u) circle (0.5pt);
\node[right=4pt] at (0,\u+0.05) {{$\mu_B^+$}};
\filldraw[color=blue, fill=white] (\thetaA,\UAkink) circle (0.5pt);
\node[right=4pt] at (\thetaA-0.02,\UAkink-0.05) {{$\mu_A^+$}};

\filldraw[color=blue, fill=white] (0,\u-\c) circle (0.5pt);
\node[right=4pt] at (0,\u-\c-0.05) {{$\mu_A^-$}};
\filldraw[color=red, fill=white] (\thetaB,\u) circle (0.5pt);
\node[right=4pt] at (\thetaB-0.03,\u-0.05) {{$\mu_B^-$}};
\end{tikzpicture}
\end{minipage}
\caption{\label{fig:transparency and welfare}\textit{Proposition \ref{prop:transparency and welfare}}. For visual clarity, the domain in both panels extends just past $\theta_B$ (not to $1$).  Left panel: the \captive{}'s equilibrium welfare. Payoff $u^{cap}_A(\cdot)$ is red, $u^{cap}_B(\cdot)$ blue and posteriors under positive prominence are labeled $\{\mu_A^+,\mu_B^+\}$ and joined together by the dashed line, representing payoff from positive prominence. Posteriors under negative prominence, labeled $\{\mu_A^-,\mu_B^-\}$, are also joined together by the dashed line. For prior beliefs below $\theta_A$, the purple curve represents the \captive{}'s equilibrium payoff. At high priors, the payoff under transparency coincides with $u^{cap}_A(\cdot)$. Right panel: the analogous construction for the \searcher{}. As drawn, the \searcher{}'s payoff with negative sorting is lower than $u_A^{se}(\cdot)$, but this ranking depends on parameters.}
\label{fig:search}
\end{figure}
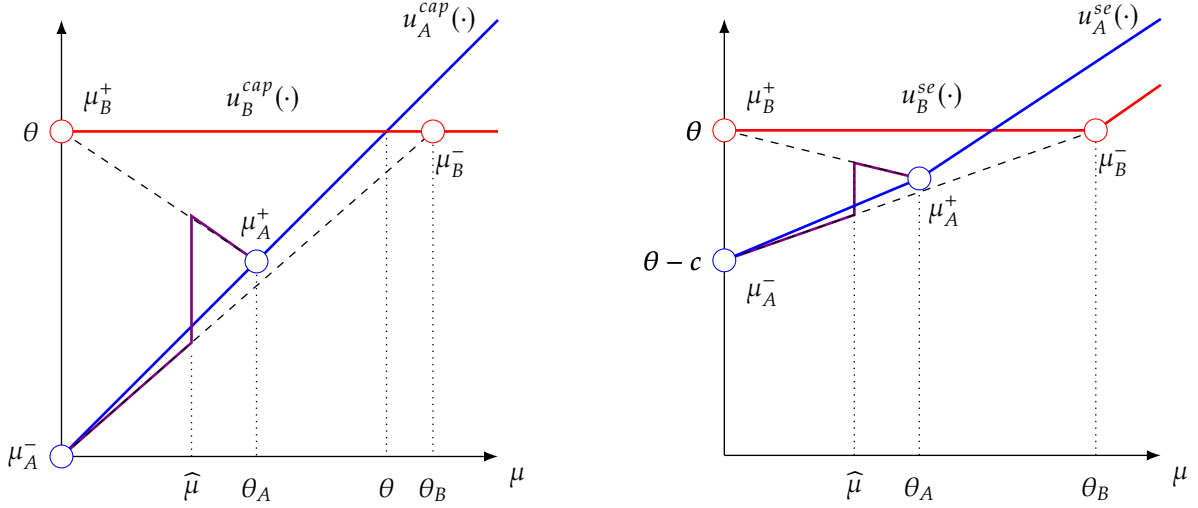

To understand the impact of transparency on the consumer, we focus on $\prior<\theta_A$, and compare each type's equilibrium payoff under transparency with her payoff in the unique equilibrium under opacity, in which $A$ is always prominent. Such an equilibrium delivers payoff $u^j_A(\prior)$ to the type $j\in\{cap,se\}$ consumer. 

Examining Figure \ref{fig:transparency and welfare} reveals that the \captive{}'s equilibrium payoff under transparency is strictly smaller than $u^{cap}_A(\cdot)$ at low priors ($\prior<\widehat{\mu}$) where the platform uses negative sorting. Simultaneously, it is strictly larger at moderate priors ($\prior\in(\widehat{\mu},\theta_A)$) where the platform uses the positive sorting. Both of these hold generally: the first follows from $\theta_B>\out$, the second from $c>0$. 

For the \searcher{}, the argument is nearly identical for moderate priors. As is evident from the right panel of Figure \ref{fig:transparency and welfare}, the \searcher{}'s payoff from positive sorting exceeds $u_A^{se}(\cdot)$, and thus transparency is beneficial for moderate priors ($\mu\in(\widehat{\mu},\theta)$). However, the analysis of negative sorting is more involved. In particular, the value  $u_A^{se}(\theta_A)$ depends on parameters, and in the proof of Proposition \ref{prop:transparency and welfare} we show that it can lie above or below the payoff of negative sorting, 
so that either ranking is possible for the \searcher{}.

\begin{proposition}[\textbf{Transparency and Welfare}]
\label{prop:transparency and welfare}
Compare the payoffs of the platform, the \captive{}, and the \searcher{} in the unique  equilibrium under transparency with their payoffs in the platform-preferred equilibrium under opacity (\(A\)-always).
\begin{itemize}
\item At low priors, \(\prior\in(0,\widehat{\mu})\), the platform is strictly better off under transparency, while the \captive{} is strictly worse off. The \searcher{} may be better or worse off depending on parameters.

\item At moderate priors, \(\prior\in(\widehat{\mu},\theta_A)\), the platform, the \captive{}, and the \searcher{} are all strictly better off under transparency.

\end{itemize}
\end{proposition}

To see the broad idea, note that positive sorting deters search, and both types of consumer buy the prominent product with high probability. Furthermore, with positive sorting prominence conveys good news about product fit. Thus, the consumer is steered toward the product that is the better match. In particular, Figure \ref{fig:transparency and welfare} reveals that the gains from positive sorting for both types arise when $B$ is prominent and revealed to be the better match, which allows it to be purchased without incurring a search cost. In other words, if the payoff when $B$ is prominent were evaluated on the blue curve instead of red, both types would be indifferent between opacity and transparency at moderate priors.

In contrast, negative sorting encourages the \searcher{} to continue by conveying bad news about the prominent product. The consumer is therefore steered toward a product that she would rather not purchase immediately, which imposes losses on the \captive{} for whom continued search is prohibitive. In particular, when $B$ is prominent, the \captive{} is forced to buy $B$, even though she strictly prefers $A$ at belief $\mu_B^-=\theta_B$. As can be seen in the left panel of Figure \ref{fig:transparency and welfare}, this distortion is the source of the \captive{}'s welfare harm. If her payoff at $\mu_B^-$  were evaluated on the blue curve, she would be indifferent between transparency and opacity at low priors. For the \searcher{}, the welfare impact of negative sorting is more involved. Unlike the \captive{} who is effectively locked into the prominent product, the \searcher{} is more flexible, which allows her to take advantage of the information revealed by negative sorting. Thus, her ranking depends on the specific frictions that she faces. As we show in the proof, the \searcher{} benefits from transparency if and only if the information friction is large relative to the cost of search ($r<r^*(c,\out)$). In this case, the information generated by negative sorting is valuable enough to offset the search cost. Otherwise, the consumer gains substantial information from encountering $A$, and the additional information generated by the algorithm is of little value.

 \section{The Wrong Kind of Transparency}\label{sec:wrong kind of transparency}

In the previous section, we showed that transparency allows the display algorithm to influence the consumer's product choice both by manipulating the search order, and by using prominence to convey information. As a result, transparency can generate a strict Pareto improvement relative to the platform-preferred equilibrium under opacity. Using the algorithm for both purposes simultaneously, however, may seem to place a substantial burden on the consumer. In particular, the consumer may have to buy a product that she knows is mismatched, or incur a search cost to avoid doing so, because the platform would like to convey information and can only use product prominence to do so. The consumer and platform might both benefit if the platform could transmit information about the new product independently of prominence (directly).

\paragraph{Recommendation and Display Algorithms.} To explore this possibility, consider an algorithm that recommends a product to the consumer conditional on the state. In particular, interpret message $m\in\{a,b\}$ as a recommendation of the  corresponding product. Recommendation algorithm $\Sigma\equiv(\sigma_0,\sigma_1)$, recommends product $A$ ($m=a$) in state $\omega$ with probability $\sigma_\omega$. We denote the beliefs induced by the recommendation algorithm as $(\mu_a,\mu_b)$ with the convention that $\mu_b\leq \mu_a$, i.e., $m$ is a recommendation. Because both the consumer and platform observe the recommendation $m$, the display algorithm can condition on the realized recommendation. Thus, with a recommendation system, the platform designs a pair of display algorithms, $\{\Pi^a,\Pi^b\}$. When product $m$ is recommended by the recommendation algorithm, prominence is determined according to $\Pi^m$. 

To streamline the exposition, we focus on the case where  the recommendation algorithm and the display algorithms $\{\Sigma,\Pi^a,\Pi^b\}$ are jointly designed in stage (i) and are transparent. This specification gives the platform maximum flexibility: it can convey information using recommendations, prominence, or any combination of the two.

\paragraph{Joint Design.} Consider a recommendation algorithm $\Sigma=(\sigma_0,\sigma_1)$, which induces posterior belief $\mu_m$ when recommendation $m\in\{a,b\}$ is transmitted. Conditional on each message $m\in\{a,b\}$ and its induced belief $\mu_m$, the platform designs an observable display algorithm $\Pi^m$, as in Section \ref{sec:transparency}. By implication, conditional on recommendation $m$, the highest profit that the platform can attain is $v^*(\mu_m)$ (see Figure \ref{fig:equilibrium under transparency}). Therefore, the platform designs its recommendation algorithm $\Sigma$, anticipating that profit $v^*(\mu_m)$ for each realization of the posterior belief. As is well known \citep{Kamenica_Gentzkow_2011_AER}, by designing its recommendation algorithm $\Sigma$, the platform can attain any profit in the convex hull of $v^*(\cdot)$, and the highest attainable profit is its concave envelope, i.e., $V^*(\cdot)\equiv con(v^*(\cdot))$.\footnote{The logic is similar to transparent algorithm design, except that the profit functions $v_A(\cdot)$ and $v_B(\cdot)$ are replaced by  $v^*(\cdot)$. In contrast to Section \ref{sec:transparency}, the platform's profit depends on the realization of the recommendation $m$ only through the posterior belief.} The beliefs at which the profit function $v^*(\cdot)$ and its concave envelope $V^*(\cdot)$ coincide are induced by the optimal recommendation system.

\begin{figure}
\begin{minipage}{0.5\textwidth}
\begin{tikzpicture}[scale=0.85,  >=Latex,]

\draw[->, line width=0.5pt] (0,-1) -- (0,5);
\fill (7,-1) node[right] {\footnotesize{$\mu$}};
\draw[->, line width=0.5pt] (0,-1) -- (7,-1);

\draw[line width=1pt,blue] (0,0)--(3,1);
\draw[line width=1pt,blue, dashed] (3,1)--(3,3.18);
\draw[line width=1pt,blue] (3,3.18)--(7,4.5);
 \draw[line width=0.5pt,black,dotted] (3,-1)--(3,1);
\fill (0.6,2) node[right] {\footnotesize{$v^*(\cdot)$}};
\fill (4,-1) node[below] {\footnotesize{$\theta_B$}};

\fill (3,-1) node[below] {\footnotesize{$\theta_A$}};

\draw[line width=1pt,red] (0,-1)--(4,-1);

\draw[line width=1pt,red] (4,2.5)--(7,3.5);
\draw[line width=1pt,red, dashed] (4,-1)--(4,2.5);

\filldraw[color=red, fill=red] (4,2.5) circle (3pt);
\filldraw[color=blue, fill=blue] (3,3.18) circle (3pt);

\draw[line width=1pt,violet] (0,0)--(1.33,0.85)--(3,3.18);
\draw[line width=1pt,violet] (3,3.18)--(7,4.5);
\draw[line width=0.5pt, dashed] (1.33,0.85) -- (4,2.5);
\draw[line width=0.5pt, dashed] (0,-1) -- (1.33,0.85);
\fill (0,0) node[left] {\footnotesize{$\mu_{A}^{-}$}};
\fill (3,3.5) node[above] {\footnotesize{$\mu_{A}^{+}$}};
\fill (4,2.5) node[above] {\footnotesize{$\mu_{B}^{-}$}};
\fill (0,-1) node[left] {\footnotesize{$\mu_{B}^{+}$}};
\filldraw[color=red, fill=white] (4,2.5) circle (3pt);
\filldraw[color=red, fill=white] (0,-1) circle (3pt);
\filldraw[color=blue, fill=white] (3,3.18) circle (3pt);
\filldraw[color=blue, fill=white] (0,0) circle (3pt);

\draw[line width=0.5pt, dotted] (1.33,-1) -- (1.33,0.85); 
\fill (1.33,-1) node[below] {\footnotesize{$\hat{\mu}$}};

\end{tikzpicture}
\end{minipage}
\begin{minipage}{0.5\textwidth}
\begin{tikzpicture}[scale=0.85,  >=Latex,]

\draw[->, line width=0.5pt] (0,-1) -- (0,5);
\fill (7,-1) node[right] {\footnotesize{$\mu$}};
\draw[->, line width=0.5pt] (0,-1) -- (7,-1);

\draw[line width=1pt,blue] (0,0)--(3,1);
\draw[line width=1pt,blue, dashed] (3,1)--(3,3.18);
\draw[line width=1pt,blue] (3,3.18)--(7,4.5);
\fill (0.6,2) node[right] {\footnotesize{$V^*(\cdot)$}};
\fill (4,-1) node[below] {\footnotesize{$\theta_B$}};

\fill (3,-1) node[below] {\footnotesize{$\theta_A$}};

\draw[line width=1pt,red] (0,-1)--(4,-1);

\draw[line width=1pt,red] (4,2.5)--(7,3.5);
\draw[line width=1pt,red, dashed] (4,-1)--(4,2.5);

\filldraw[color=red, fill=red] (4,2.5) circle (3pt);
\filldraw[color=blue, fill=blue] (3,3.18) circle (3pt);

\draw[line width=0.5pt,black,dotted] (3,-1)--(3,1);

\draw[line width=1pt,violet] (0,0)--(1.33,0.85)--(3,3.18);
\draw[line width=1pt,violet] (3,3.18)--(7,4.5);
\draw[line width=0.5pt, dashed] (1.33,0.85) -- (4,2.5);
\draw[line width=0.5pt, dashed] (0,-1) -- (1.33,0.85);
 \fill (0,0) node[left] {\footnotesize{$\mu_{b}^{*}$}};
 \fill (3,3.3) node[above] {\footnotesize{$\mu_{a}^{*}$}};
\draw[line width=1pt,green!60!black] (0,0)--(3,3.18)--(7,4.5);

\filldraw[color=red, fill=white] (4,2.5) circle (3pt);
\filldraw[color=red, fill=white] (0,-1) circle (3pt);
\filldraw[color=blue, fill=white] (3,3.18) circle (3pt);
\filldraw[color=blue, fill=white] (0,0) circle (3pt);

\draw[line width=0.5pt, dotted] (1.33,-1) -- (1.33,0.85); 
\fill (1.33,-1) node[below] {\footnotesize{$\hat{\mu}$}};

\end{tikzpicture}
\end{minipage}
\caption{\label{fig:joint design} \textit{Proposition \ref{prop:equilibrium joint design}}. Left panel: the highest attainable profit with a transparent display algorithm $v^*(\cdot)$. Right panel: the highest attainable profit under joint design $V^*(\cdot)$, the concave envelope of $v^*(\cdot)$. At low beliefs, the induced posteriors are $\{\mu_a=\theta_A,\mu_b=0\}$.  }
\end{figure}
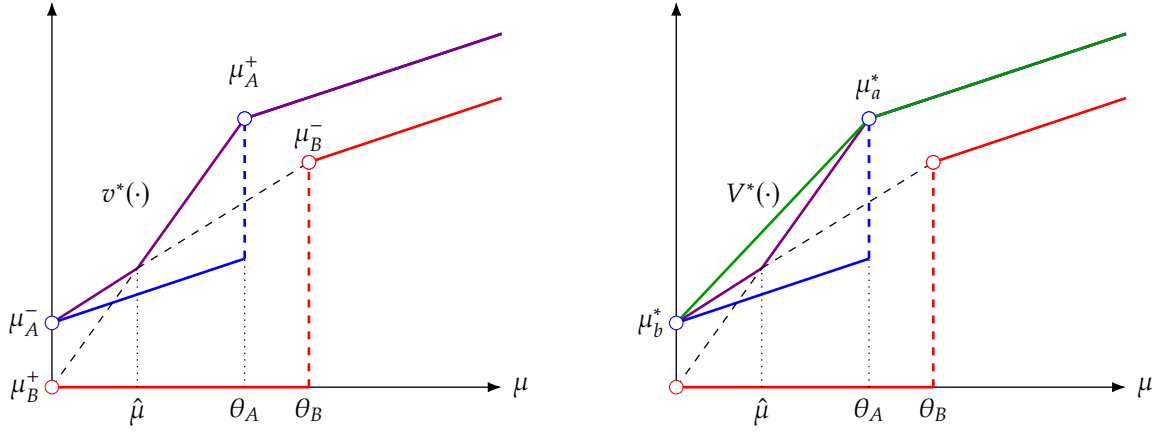

\paragraph{Equilibrium Joint Design.} The following characterization of the equilibrium recommendation and display algorithms follows from Figure \ref{fig:joint design}. Recall that the platform's highest possible profit is attained under transparency when $\prior\geq \theta_A$. For such priors, joint design cannot possibly be a strict improvement. We therefore focus on $\prior\in(0,\theta_A)$, where the recommendation algorithm could possibly do so.

\begin{proposition}[\textbf{Equilibrium Joint Design}]\label{prop:equilibrium joint design} When the joint algorithm $(\Sigma,\Pi^a,\Pi^b)$ is transparent and $\prior\in(0,\theta_A)$, a unique equilibrium exists. The equilibrium recommendation algorithm $\Sigma$ induces beliefs $\{\mu_a=\theta_A,\,\mu_b=0\}$ and the display algorithms $(\Pi^a,\Pi^b)$ always make $A$ prominent.
\end{proposition}

Proposition \ref{prop:equilibrium joint design} reveals that the platform strictly benefits by decoupling information provision from prominence (for the relevant priors). Indeed, the recommendation algorithm conveys information to the consumer ($\mu_a\neq\prior$ and $\mu_b\neq\prior$), while the display algorithm is uninformative and always makes $A$ prominent. In other words, all information is transmitted through the recommendation algorithm, while the display algorithms  maximally steer the consumer towards $A$. Consequently, the same profit can be attained if the recommendation algorithm is transparent but the display algorithms are opaque. In contrast, it can be shown that if the recommendation algorithm is opaque, then recommendations do not substantively change the equilibrium, regardless of whether the display algorithm is transparent or opaque.%

As is clear from Figure \ref{fig:joint design}, the platform strictly benefits from a transparent recommendation algorithm for the relevant prior beliefs ($\prior\in(0,\theta_A)$). The recommendation algorithm allows the platform to exploit the search friction to the maximum extent while also conveying information to the \searcher{}. Indeed, it can always make $A$ prominent to steer aggressively, while using the recommendation to  persuade the \searcher{} to stop. 

Though the platform benefits when the recommendation algorithm is transparent, both types of consumers might be harmed. When such an algorithm is available, $A$ is always prominent, and the \captive{} always buys it, attaining payoff $u_A^{cap}(\prior)$. Furthermore, the posterior beliefs induced by the recommendation system are $(\mu_a=\theta_A,\mu_b=0)$. Because the \searcher{}'s payoff $u_A^{se}(\cdot)$ is linear below $\theta_A$, her expected payoff is as if $A$ is always prominent and the recommendation is uninformative, $u_A^{se}(\prior)$ (see Figure \ref{fig:transparency and welfare}, right panel). Intuitively, the equilibrium recommendation system never provides information that is sufficiently informative to strictly overturn the search strategy that she would adopt at the prior belief. Thus, for each type of consumer, $j\in\{se,cap\}$, the equilibrium payoff under joint design is $u_A^j(\cdot)$, as in the equilibrium under opacity. Combining these observations with Proposition \ref{prop:transparency and welfare}, we find that transparency of the recommendation algorithm harms the consumer when the prior belief is moderate. When the prior belief is low, it benefits the \captive{}, and may either benefit or harm the \searcher{}.

 \begin{proposition}[\textbf{The Wrong Kind of Transparency}]\label{prop:wrong kind of transparency} Suppose that $\prior\in(0,\theta_A)$ and the display algorithm is transparent. Compare the equilibrium payoffs of the platform, \captive{}, and \searcher{} when a transparent recommendation algorithm exists (Proposition \ref{prop:equilibrium joint design}) and when it does not (Proposition \ref{prop:equilibrium under transparency}). 
\begin{itemize}
\item The platform is strictly better off with transparent recommendations.
\item At low priors ($\prior\in(0,\widehat{\mu}$) the  \captive{} is strictly better off with transparent recommendations. The \searcher{} may be better or worse off depending on parameters.
\item At moderate priors ($\prior\in(\widehat{\mu},\theta_A)$) the \captive{} and \searcher{} are strictly worse off with transparent recommendations.
\end{itemize}
\end{proposition}

Proposition \ref{prop:wrong kind of transparency} gives insight into the implications of different kinds of transparency. When the prior belief is moderate, a transparency requirement that ensures consumers understand how prominence is formed and what it means,  simultaneously increases the welfare of the platform, \searcher{}, and \captive{}. In this case, extending or shifting the transparency requirement to the recommendation algorithm further benefits the platform, while reversing the consumer's welfare gains. In other words, if the transition from an opaque to a transparent display algorithm is beneficial for the consumer, then extending or shifting transparency to the recommendation algorithm is harmful. At low prior beliefs, in contrast, a transparent display algorithm may harm one or both types of consumer. Requiring that recommendations are transparent undoes this harm, while generating even more profit for the platform.

\section{Extensions}\label{sec:extensions}

We consider two extensions of our main model. First, we study how prominence is monetized. Second, we consider more than two more products.

\subsection{Monetizing Prominence}\label{sec:monetizing}
The main model treats the commission structure as given and studies the platform's optimal display algorithm. In this section, we consider how prominence is monetized.

\paragraph{Equilibrium Commission Formation.} We augment the game with an initial stage (before design) in which each firm $i\in\{A,B\}$ simultaneously offers a per-sale commission $k_i$, paid to the platform when its product sells.  For now, we treat the probability of sale in reduced form. In particular, \(p_i^H\) denotes the probability that firm \(i\in\{A,B\}\) sells its product in the continuation game when it offers the higher commission, \(k_i>k_{j}\), and \(p_i^L\) denotes the probability that firm \(i\) sells in the continuation game when it offers the lower commission, \(k_i<k_{j}\).\footnote{If commissions are equal $k_i=k_j$, the platform is indifferent, and it breaks ties in favor of the firm with higher effective valuation ($\eta_i$). The tie-breaking rule does not affect the characterization.} Because the consumer always buys one of the products, $p_i^H+p_j^L=1$. Each firm is more likely to sell when its commission is high than when it is low, $0\leq p_i^L<p_i^H\leq 1$. The difference between these probabilities is the same for both firms.\footnote{We have $p_i^H=1-p_j^L$ and $p_i^L=1-p_j^H$. Hence, $p_i^H-p_i^L=1-p_j^L-(1-p_j^H)=p_j^H-p_j^L$.} Thus, let
\[
\Delta\equiv p_i^H-p_i^L=p_j^H-p_j^L.
\]
To highlight the role of the display algorithm in shaping equilibrium commissions, we focus on the case where both firms have the same margin $m$ per-sale. Thus, if firm $i$ offers commission $k_i$ and sells its product, then its net profit is $m-k_i$.

Given the discontinuity in the probability of sale when a firm's commission switches from being smaller to larger, we expect an equilibrium in mixed strategies. If firm $i$ anticipates that $j$'s commission is realized from cdf $G_j$, then $i$'s expected profit from commission $k$ is 
\[
\pi_i(k)=(m-k)\left(p_i^L+G_{j}(k)\Delta\right).
\]
A higher commission reduces the firm's net surplus on each sale, but increases the probability that its commission is higher, which increases its probability of sale downstream. A firm must pay the commission even if its commission is smaller and its sale probability is smaller. This all-pay element creates some similarities to the standard all-pay auction. Unlike the standard auction, however, a firm's total payment depends both on the commission it offers and on the favorable or unfavorable treatment it receives. This distinction  precludes a payoff-equivalent transformation between settings.

The characterization of equilibrium commissions depends on two parameters which represent each firm's willingness to exchange commissions for favorable treatment,
\[
\eta_i\equiv \frac{m\Delta}{p_i^H}.
\]
To see the interpretation, suppose for a moment that firm $i$ has only two options. Either it sells with the low probability ($p_i^L$) for free, or it offers   commission $\tilde k>0$ and sells with the high probability ($p_i^H$). When $\tilde k=\eta_i$, firm $i$ is indifferent. In other words, if $p_i^L$ is free, then $\eta_i$ is the highest commission that firm $i$ would offer to obtain $p_i^H$ (i.e., $m\Delta=p_i^H\eta_i$). We therefore refer to $\eta_i$ as firm $i$'s \textit{effective valuation}.

The following lemma is proved in the Appendix.

\begin{lemma}[\textbf{Competitive Commissions}]\label{lem:competitive commissions}
 The reduced form of the commission stage has a unique equilibrium, which is in mixed strategies. Suppose that firm $i\in\{A,B\}$ has a higher effective valuation than firm $j\neq i$, i.e., $\eta_i>\eta_j$. Firm $i$ randomizes continuously on \([0,\eta_j]\) with CDF
\[
G_i(k)
=
\left(\frac{m-\eta_j}{\eta_j}\right)\frac{k}{m-k},
\qquad k\in[0,\eta_j].
\]
Firm $j$ bids zero with probability \(1-\eta_j/\eta_i\). With the remaining probability \(\eta_j/\eta_i\), it draws from the same continuous distribution as firm $i$,
\[
G_j(k)
=
1-\frac{\eta_j}{\eta_i}
+
\frac{\eta_j}{\eta_i}G_i(k),
\qquad k\in[0,\eta_j].
\]

 Equilibrium payoffs are
\(
U_i=(m-\eta_j)p_i^H\),
\(U_{j}=mp_j^L\) and \( V=\eta_jp_i^H\)

\end{lemma}

Equilibrium commissions have certain structural features reminiscent of the standard all-pay auction. In particular, the firm with the higher effective value  offers a commission that is continuously distributed from zero up to the other firm's effective value. The firm with the smaller effective value sometimes offers commission zero, but otherwise mimics its rival's distribution. Furthermore, the probabilities of sale arising from favorable or unfavorable treatment affect equilibrium commissions only through the effective valuations, and rents are similarly dissipated, pinning the profit of the firm with lower effective value to its outside option.\footnote{Unlike the standard auction, a firm can sell with probability $p_i^L$, even if it offers commission zero.} The crucial difference is in the mixing distributions, which are non-linear.

The incentive to offer commissions is determined by the difference in sale probabilities from winning and losing, $\Delta=p_i^H-p_i^L$, which is identical for both firms. This difference strictly increases when either $p_A^H$ or $p_B^H$ increases, and in turn results in a higher expected profit, inclusive of commission formation and sale probabilities, as shown in Lemma \ref{lem:favored treatment}.

\begin{lemma}[\textbf{Favored Treatment and Profit}]\label{lem:favored treatment} In the commission stage, the platform's equilibrium
profit is strictly increasing in each firm's probability of sale under favored treatment $(p_A^H,p_B^H)$.
\end{lemma}

\paragraph{Commissions and Downstream Algorithms.} After the commissions are set initially, the platform designs its display algorithm. Because the consumer always buys exactly one product, the platform's expected revenue under any continuation algorithm is
\[
k_B+(k_A-k_B)\Pr(A\text{ is sold})=k_A+(k_B-k_A)\Pr(B\text{ is sold}).
\]
It follows that only the difference in commissions  matters for the platform's subsequent design decision. Thus, the platform would like to maximize sales of the product with higher commission. In other words, in equilibrium under transparency, $p^H_i$ is the highest  attainable probability of sale for product $i$ given the prior belief. In turn, Lemma \ref{lem:favored treatment} implies that such a choice of downstream algorithms, not only maximizes the probability that the high commission product is sold \textit{after} commissions are set, but it also maximizes the platform's expected profit at the commission stage, which accounts for the effect of the downstream algorithm on the commissions themselves. 

An immediate implication is that the platform has a higher equilibrium profit if it lets the firms offer commissions and freely selects the algorithm, than if it commits to always assign prominence to the high commission firm. By making such a commitment, the platform effectively sets $p_A^H=(1-h)v_A(\prior)$, which is weakly smaller than the highest $p_A^H$ that can be attained under transparency, $(1-h)v^*(\prior)$ (and the same is true
for $B$). A similar argument also implies that the platform strictly benefits from transparency (and from joint design) whenever it strictly increases the probability of sale for either product.

\begin{proposition}[\textbf{Monetization}]\label{prop:monetization} In equilibrium with endogenous commissions,
\begin{itemize}
\item firms offer commissions as described in Lemma \ref{lem:competitive commissions}, and the platform designs an algorithm that maximizes the probability of sale for the high commission firm. 

\item the platform's profit is weakly higher than if it commits always to assign the prominent slot to the high commission firm, and it is strictly higher if and only if $\prior\notin[\theta_A,\theta_B]$.

\item the platform's profit is weakly higher under transparency than opacity, and it is strictly higher if and only if $\prior\notin[\theta_A,\theta_B]$. 
\end{itemize}
\end{proposition}

\subsection{More Products}\label{sec:more products}

In this section we introduce an additional product, examining how a product's payoff and commission mutually affect the platform's incentive to make it prominent. Throughout this section, we focus on transparent display algorithms, as in Section \ref{sec:transparency}.

\paragraph{A Dominated Product.} Suppose that in addition to products $A$ and $B$, the platform also carries a product $D$, which is worse than $B$ for both the platform and consumer. In particular, the consumer's payoff of product $D$ is known to be $\outI\in(\out-c,\out)$. Thus, product $D$ offers the consumer a smaller payoff than $B$, but it is not so low that the consumer prefers to search for $B$ rather than buy $D$ when $D$ is prominent. To streamline the exposition, we assume that product $D$ also satisfies the regularity conditions. The commission on product $D$ is also smaller than on $B$. Recalling that the commission on $B$ is normalized to zero, the commission on $D$ is $-k_D<0$. Thus, $k_D$ is the magnitude of the commission gap.

\paragraph{Timing.} The timing of the game up to stage (iii) is identical, except that the algorithm can also make $D$ prominent. If the \searcher{} decides not to buy the prominent product during the encounter (stage (iii)), then the continuation is a standard sequential search with perfect recall and search cost $c$, in which the prominent product has already been encountered and can be freely recalled. Thus, if she continues past the prominent product, the consumer decides which products to encounter and in what order. As before, search is prohibitively costly for the \captive{}, who always buys the prominent product.

\paragraph{Search Incentives and Profit.} The addition of product $D$ to the platform's offerings does not change the consumer's search behavior or the platform's expected profit when $A$ or $B$ is prominent. In either case, the payoff of buying $B$ (either $\out-c$ or $\out$) is strictly larger than the payoff of $D$ ($\outI-c$). When $D$ is prominent, however, it can be purchased with no search cost, which must be paid to buy $B$. Consequently $D$ is a better choice than $B$ ($\outI>\out-c$). In other words, when $D$ is prominent, the consumer's choice is between $A$ and $D$, and the consumer's optimal search strategy is analogous to the second point of Lemma \ref{lem:searcher strategy}, with $\outI$ replacing $\out$. Following similar logic to Lemma \ref{lem:transparent profit}, when $D$ is prominent and the common belief is $\mu$, the platform's expected profit (normalized by $1-h$) is 
\[v_D(\mu)=\underbrace{-fk_D}_{\text{\captive{}}}-\underbrace{k_D+(r\mu+1-r)(1+k_D)\mathbbm{1}\{\mu\geq\theta_D\}}_{\text{\searcher{}}}\quad\text{and}\quad\theta_D\equiv\frac{\outI(1-r)+c}{1-r\outI}\in(\theta_A,\theta_B)\]
The consumer always purchases $D$ if she is a \captive{} ($-fk_D$). If she is a \searcher{}, she buys $D$ unless (i) the belief is above threshold $\theta_D$ so she continues, and (ii) her signal realization is favorable or inconclusive. If she purchases $A$, then the commission increases by $1+k_D$ from $-k_D$ to 1. It is important to point out that $\theta_D\in(\theta_A,\theta_B)$: the consumer has a lower payoff of $D$ than $B$ and is therefore more willing to search when $D$ is prominent.

The key relationships between $B$ and $D$ is illustrated graphically in the left panel of Figure \ref{fig:dominated product}. Because its commission is lower, the graph of $v_D(\cdot)$ (in orange) is below $v_B(\cdot)$ at the extremes. At the same time, because it offers the consumer a smaller payoff, search threshold $\theta_D<\theta_B$, so the jump in $v_D(\cdot)$ is to the left, and $v_D(\cdot)$ is above $v_B(\cdot)$ in the middle.

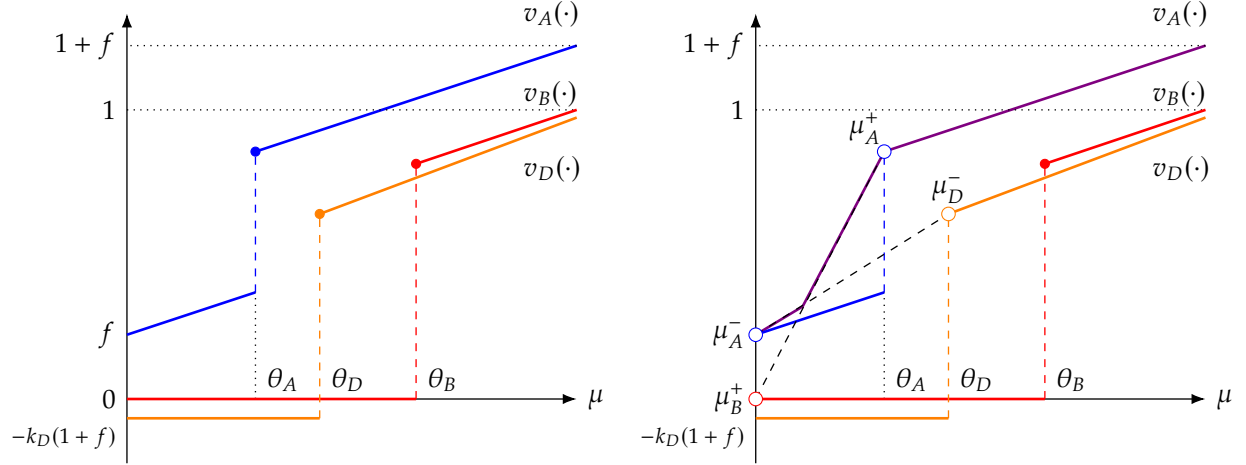
\begin{figure}
\begin{minipage}{0.5\textwidth}
\begin{tikzpicture}[scale=0.85,  >=Latex]

\draw[->,line width=0.5pt] (0,-2) -- (0,5);
\fill (7,-1) node[right] {\footnotesize{$\mu$}};
\draw[->,line width=0.5pt] (0,-1) -- (7,-1);

\draw[line width=1pt,blue] (0,0)--(2,0.66);
\draw[line width=0.5pt,blue,dashed] (2,0.66)--(2,2.85);
 \draw[line width=1pt,blue] (2,2.85)--(7,4.5);
\draw[line width=0.5pt,black,dotted] (2,-1)--(2,0.66);
\fill (6,5) node[right] {\footnotesize{${v}_A(\cdot)$}};

\fill (0,0) node[left] {\footnotesize{$f$}};

\draw[line width=0.5pt,black,dotted] (7,4.5)--(0,4.5);
\fill (0,4.5) node[left] {\footnotesize{$1+f$}};
\fill (0,3.5) node[left] {\footnotesize{$1$}};
\fill (0,-1) node[left] {\footnotesize{$0$}};
\draw[line width=0.5pt,black,dotted] (0,3.5)--(7,3.5);

\draw[line width=1pt,red] (0,-1)--(4.5,-1);
\draw[line width=1pt,red] (4.5,2.66)--(7,3.5);
\draw[line width=0.5pt,red,dashed] (4.5,-1)--(4.5,2.66);
\fill (6,3.75) node[right] {\footnotesize{$v_B(\cdot)$}};

\def\offY{0.3}
\def\offX{1.5}

\pgfmathsetmacro{\offYY}{0.33*(7-(4.5-\offX))}
\draw[line width=1pt,orange] (0,-1-\offY)--(4.5-\offX,-1-\offY);
\draw[line width=1pt,orange] (4.5-\offX,3.5-\offY - \offYY)--(7,3.5-\offY+0.18);
\draw[line width=0.5pt,orange,dashed] (4.5-\offX,-1-\offY)--(4.5-\offX,3.5-\offY - \offYY);
\fill (6,2.7-\offY+0.18) node[right] 
{\footnotesize{$v_D(\cdot)$}};

\fill (4.5-\offX,-0.7) node[right] {\footnotesize{$\theta_D$}};

\fill (4.5,-0.7) node[right] {\footnotesize{$\theta_B$}};
\fill (2,-0.7) node[right] {\footnotesize{$\theta_A$}};

\fill (0,-1-\offY-0.3) node[left] {\scriptsize{$-k_D(1+f)$}};

\filldraw[color=red, fill=red] (4.5,2.66) circle (2pt);
\filldraw[color=blue, fill=blue] (2,2.85) circle (2pt);
\filldraw[color=orange, fill=orange] (4.5-\offX,3.5-\offY - \offYY) circle (2pt);

\end{tikzpicture}
\end{minipage}
\begin{minipage}{0.5\textwidth}
\begin{tikzpicture}[scale=0.85,  >=Latex]

\draw[->,line width=0.5pt] (0,-2) -- (0,5);
\fill (7,-1) node[right] {\footnotesize{$\mu$}};
\draw[->,line width=0.5pt] (0,-1) -- (7,-1);

\filldraw[color=red, fill=red] (4.5,2.66) circle (2pt);
\draw[line width=1pt,violet] (0,0)--(0.73,0.44)--(2,2.85);

\draw[line width=1pt,blue] (0,0)--(2,0.66);
\draw[line width=0.5pt,blue,dashed] (2,0.66)--(2,2.85);
\draw[line width=1pt,violet] (2,2.85)--(7,4.5);
\draw[line width=0.5pt,black,dotted] (2,-1)--(2,0.66);
\fill (6,5) node[right] {\footnotesize{${v}_A(\cdot)$}};

\fill (0,0) node[left] {\footnotesize{$\mu_{A}^{-}$}};
\fill (1.7,2.8) node[above] {\footnotesize{$\mu_{A}^{+}$}};

\draw[line width=0.5pt,black,dotted] (7,4.5)--(0,4.5);
\fill (0,4.5) node[left] {\footnotesize{$1+f$}};
\fill (0,3.5) node[left] {\footnotesize{$1$}};
\fill (0,-1) node[left] {\footnotesize{$\mu_{B}^{+}$}};
\draw[line width=0.5pt,black,dotted] (0,3.5)--(7,3.5);

\draw[line width=1pt,red] (0,-1)--(4.5,-1);
\draw[line width=1pt,red] (4.5,2.66)--(7,3.5);
\draw[line width=0.5pt,red,dashed] (4.5,-1)--(4.5,2.66);
\fill (6,3.75) node[right] {\footnotesize{$v_B(\cdot)$}};

\def\offY{0.3}
\def\offX{1.5}

\pgfmathsetmacro{\offYY}{0.33*(7-(4.5-\offX))}
\draw[line width=1pt,orange] (0,-1-\offY)--(4.5-\offX,-1-\offY);
\draw[line width=1pt,orange] (4.5-\offX,3.5-\offY - \offYY)--(7,3.5-\offY+.18);
\draw[line width=0.5pt,orange,dashed] (4.5-\offX,-1-\offY)--(4.5-\offX,3.5-\offY - \offYY);
\fill (6,2.7-\offY+0.18) node[right] 
{\footnotesize{$v_D(\cdot)$}};

\fill (4.5-\offX,-0.7) node[right] {\footnotesize{$\theta_D$}};
\fill (4.5,-0.7) node[right] {\footnotesize{$\theta_B$}};
\fill (2,-0.7) node[right] {\footnotesize{$\theta_A$}};

\fill (0,-1-\offY-0.3) node[left] {\scriptsize{$-k_D(1+f)$}};

\draw[line width=0.5pt, dashed] (0,-1) -- (2,2.85);
\draw[line width=0.5pt, dashed] (0,0) -- (4.5-\offX,3.5-\offY - \offYY);

\filldraw[color=red, fill=white] (0,-1) circle (3pt);
\filldraw[color=blue, fill=white] (2,2.85) circle (3pt);

\filldraw[color=blue, fill=white] (0,0) circle (3pt);
\filldraw[color=orange, fill=white] (4.5-\offX,3.5-\offY - \offYY) circle (3pt);
\fill (4.5-\offX,3.5-\offY - \offYY) node[above] {\footnotesize{$\mu_{D}^{-}$}};

\end{tikzpicture}\end{minipage}
\caption{\label{fig:dominated product}\textit{A Dominated Product}. Left Panel: the platform's expected profit functions when product $i\in\{A,B,D\}$ is prominent. Note that the slope of $v_D(\cdot)$ above $\theta_D$ is greater than the slope of $v_B(\cdot)$ above $\theta_B$ ($r(1+k_D)>r$). 
 Right panel: $k_D$ sufficiently small that it is optimal to replace $B$ with $D$ in the negative sorting algorithm. }
\end{figure}

\paragraph{Equilibrium Algorithms With a Dominated Product.} While it might be plausible that a dominated product would not be prominent in equilibrium, the right panel of Figure \ref{fig:dominated product} illustrates that this intuition is flawed. In particular, if product $D$'s commission $k_D$ is not too much smaller than $B$'s ($k_D$ small), then the negative sorting algorithm is more profitable when $D$ is displayed instead of $B$. Thus, $D$ is displayed in equilibrium with positive probability (and $B$ is not) when the prior belief is low.
\begin{proposition}[\textbf{Dominated Product}]\label{prop:dominated product} Suppose that the commission on dominated product $D$ ($-k_D$) is not too much smaller than $B$, i.e., $k_D\in(0,\bar{k})$ and the display algorithm is transparent. A $\widehat{\mu}_D\in(\widehat{\mu},\theta_A)$ exists such that the equilibrium algorithm is $(\mu_A=0,\mu_D=\theta_D)$ if and only if $\prior\in(0,\widehat{\mu}_D)$. At such priors, $B$ is never prominent.

\end{proposition}

Proposition \ref{prop:dominated product} implies that when the commission on the dominated product is not too low, the platform benefits by replacing product $B$ with $D$ when it does negative sorting. In particular, the negative sorting algorithm which displays either $A$ or $D$, $(\mu_A=0,\mu_D=\theta_D)$, is strictly more profitable than the one which displays either $A$ or $B$, $(\mu_A=0,\mu_B=\theta_B)$. Recall that with negative sorting, the consumer becomes more optimistic about $A$ when the other product is prominent, which incentivizes her to continue searching. Because $D$ is worse for the consumer, it is easier to incentivize her to continue searching when $D$ is displayed instead of $B$. Replacing $B$ with $D$ imposes a cost, however, because the platform's commission is smaller if $A$ fails to sell ($-k_D$ instead of $0$). However, if the gap in commissions is not too large, replacing $B$ with $D$ increases the platform's profit from negative sorting. Note that replacing $B$ with $D$ in positive sorting is strictly harmful.

\section{Discussion}
We summarize our main findings and situate them in a broader context.

\paragraph{Economic Context.} Prominence changes the path from a consumer's query to purchase. The search advertising industry monetizes this influence by allowing firms to pay for prominent positions or display in search results. The economic value of prominence is broadly reflected in worldwide expenditure on search adveritising, which was forecast to reach approximately \$249 billion in 2025.\footnote{\href{https://www.warc.com/en/content/article/the-big-picture-search-2025/160566}{WARC} forecasts a global search-advertising market of \$249 billion in 2025.} In turn, search rankings and sponsored placements mediate product discovery within a global retail e-commerce sector forecast at \$6.419 trillion in 2025, equivalent to 20.5 percent of worldwide retail sales.\footnote{\href{https://www.emarketer.com/chart/270813/retail-ecommerce-sales-worldwide-2022-2028-trillions-change-of-total-retail-sales}{EMARKETER} forecasts worldwide retail e-commerce sales of \$6.419 trillion and an online share of 20.5 percent of total retail sales in 2025.} Given the economic scale of these industries, even modest changes in the way ranking algorithms allocate prominence can affect economically important volumes of trade.

\paragraph{Opaque Algorithms and Privacy Protection.} Absent explicit measures to enhance transparency, consumers may have difficulty understanding why certain products are prominent and what it implies about their suitability, creating concern among regulators and policy observers. With such opacity, our analysis (Section \ref{sec:opacity}) suggests that the platform does best when the consumer anticipates that prominence is uninformative about product fit, and the platform leverages search frictions by prominently displaying the most profitable products. The platform cannot guarantee this, however: other equilibria also exist in which the consumer interprets prominent positions as bad news, reducing the platform's ability to affect consumer search and its equilibrium profit. Consequently, the platform might benefit from non-binding norms (or even explicit regulations) that discourage or prevent customization of search results based on customer attributes. 

\paragraph{Transparency Regulations.} Regulators and policy observers have pushed for greater transparency surrounding search order, prominence, and other aspects of product display. Transparency is not a technical disclosure: existing regulations require platforms and search engines to explain, in accessible language, the principal parameters that determine rankings and recommendations, why particular information is suggested, and how paid placement affects order.\footnote{Article 5 of \href{https://eur-lex.europa.eu/eli/reg/2019/1150/oj/eng}{Regulation (EU) 2019/1150} requires online search engines to provide an easily and publicly available description, drafted in plain and intelligible language, of the main parameters determining ranking and the effects of remuneration on ranking. Article 27 of the \href{https://eur-lex.europa.eu/eli/reg/2022/2065/oj/eng}{Digital Services Act} imposes similar requirements on digital platforms' recommendation systems.} In other words, transparency regulations do not require the disclosure of code or model weights, rather information that allows the consumer to understand how the platform generates product order and recommendations. 

\paragraph{Transparent Product Display and Advice.}   In Sections \ref{sec:transparency} and \ref{sec:wrong kind of transparency}, we study the effect of such transparency requirements. When algorithms governing product display are transparent, the consumer understands how prominence is assigned and what it conveys about fit. The display algorithm can therefore be used to steer and inform simultaneously, persuading the consumer to change her search behavior or product choice. With such transparency, the equilibrium algorithm sometimes conveys good news about the prominent product and deters further search, and sometimes reveals bad news that encourages it. Crucially, the platform conveys information through prominence, directly affecting the consumer's search incentives and payoffs. In other words, steering and persuasion are coupled in a single instrument, and the algorithm must simultaneously account for both.

When transparency requirements extend to algorithms that offer advice or recommend products, these two effects are decoupled. In equilibrium, the platform uses the recommendation algorithm to convey information and persuade the consumer, while prominently displaying the profitable product in order to steer the consumer toward it.

Current transparency regulations do not exclude recommendation algorithms, thereby potentially decoupling steering from persuasion.
When it is desirable to keep a general transparency requirement policy in place,
steering and persuasion can be coupled by introducing a requirement we call \textit{recommendation-display consistency}: a platform recommending a product should also be required to make the product easy to inspect and purchase. In other words, the platform cannot simultaneously make one product prominent, while recommending the consumer purchase something else. By eliminating any inconsistency between recommendations and prominence, this restriction forces information to be conveyed through product order, as in Section \ref{sec:transparency}.

\begin{figure}
\centering

\begingroup

\tikzset{
  state/.style={
    draw=black!80,
    fill=white,
    rounded corners=2pt,
    line width=0.75pt,
    minimum width=0.82cm,
    minimum height=0.58cm,
    inner xsep=3pt,
    inner ysep=1pt,
    align=center,
    font=\sffamily\small
  },
  forward/.style={
    ->,
    draw=black!85,
    line width=0.9pt
  },
  feedback/.style={
    ->,
    draw=black!85,
    line width=0.9pt,
  },
  stage/.style={
    font=\sffamily\scriptsize,
    align=center
  },
  edge note/.style={
    font=\sffamily\scriptsize,
    align=center,
    fill=white,
    inner sep=1.2pt
  }
}

\begin{minipage}[t]{0.45\linewidth}
\vspace{0pt}
\centering

\resizebox{\linewidth}{!}{%
\begin{tikzpicture}[
  x=1cm,
  y=1cm,
  >={Stealth[length=2.6mm,width=1.9mm]},
  line cap=round,
  line join=round
]

\node[state] (lp1) at (0.0,6.10) {P};
\node[state] (lp2) at (4.1,6.85) {P};
\node[state] (lp3) at (8.2,7.60) {P};

\draw[forward] (lp1) -- (lp2);
\draw[forward] (lp2) -- (lp3);

\node[state] (ls1) at (0.0,4.05) {S};
\node[state] (ls2) at (4.1,4.80) {S};
\node[state] (ls3) at (8.2,4.05) {S};

\draw[feedback] (ls1.east) -- (ls3.west);
\draw[forward] (ls1) -- (ls2);
\draw[forward] (ls2) -- (ls3);

\node[state] (lc1) at (0.0,1.90) {C};
\node[state] (lc2) at (4.1,2.65) {C};
\node[state] (lc3) at (8.2,1.90) {C};

\draw[feedback] (lc1.east) -- (lc3.west);
\draw[forward] (lc1) -- (lc2);
\draw[forward] (lc2) -- (lc3);

\node[stage] at (0.0,-0.05)
  {Opacity};

\node[stage] at (4.1,-0.05)
  {Transparent\\Display};

\node[stage] at (8.2,-0.05)
  {Transparent\\Recommendations};

\end{tikzpicture}%
}

\end{minipage}%
\hspace{0.02\linewidth}
\begin{minipage}[t]{0.45\linewidth}
\vspace{0pt}
\centering

\resizebox{\linewidth}{!}{%
\begin{tikzpicture}[
  x=1cm,
  y=1cm,
  >={Stealth[length=2.6mm,width=1.9mm]},
  line cap=round,
  line join=round
]

\node[state] (rp1) at (0.0,6.10) {P};
\node[state] (rp2) at (4.1,6.85) {P};
\node[state] (rp3) at (8.2,7.60) {P};

\draw[forward] (rp1) -- (rp2);
\draw[forward] (rp2) -- (rp3);

\node[state] (rs1) at (0.0,4.05) {S};
\node[state] (rsH) at (4.1,4.85) {S};
\node[state] (rsL) at (4.1,3.25) {S};
\node[state] (rs3) at (8.2,4.05) {S};

\draw[feedback] (rs1.east) -- (rs3.west);

\draw[forward] (rs1) -- (rsH);
\draw[forward] (rsH) -- (rs3);

\draw[forward] (rs1) -- (rsL);
\draw[forward] (rsL) -- (rs3);

\node[
  edge note,
  anchor=south
]
  at ($(rs1)!0.34!(rsH)+(0,0.43)$)
  {High info friction};

\node[
  edge note,
  anchor=north
]
  at ($(rs1)!0.34!(rsL)+(0,-0.43)$)
  {Low info friction};

\node[state] (rc1) at (0.0,1.85) {C};
\node[state] (rc2) at (4.1,0.95) {C};
\node[state] (rc3) at (8.2,1.85) {C};

\draw[feedback] (rc1.east) -- (rc3.west);
\draw[forward] (rc1) -- (rc2);
\draw[forward] (rc2) -- (rc3);

\node[stage] at (0.0,-0.05)
  {Opacity};

\node[stage] at (4.1,-0.05)
  {Transparent\\Display};

\node[stage] at (8.2,-0.05)
  {Transparent\\Recommendations};

\end{tikzpicture}%
}

\end{minipage}

\endgroup

\caption{Platform, searcher, and captive welfare under different
transparency regimes. Rightward movement strengthens transparency. Upward movement increases welfare; $P$ for platform, $S$ for \searcher{}, $C$ for \captive{}. Left Panel: Moderate priors. Right Panel: Low priors. }
\label{fig:transparency transitions}

\end{figure}
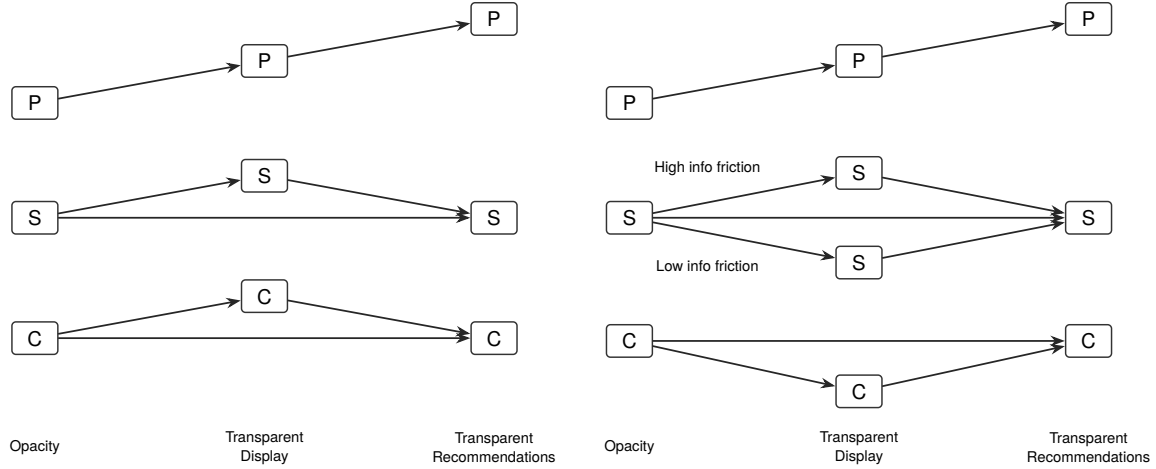

\paragraph{Welfare Impact.} The welfare impact of these different types %
of transparency is represented visually in Figure \ref{fig:transparency transitions}. While strengthening transparency always benefits the platform, the impact on consumers is mixed. When consumers tend to be skeptical of new products (low priors), transparent display is beneficial to both types. In contrast, when consumers are more optimistic (moderate priors), the welfare consequences can be mixed. If learning about new products directly on the platform is relatively easy (low information friction), then all consumers are harmed by transparent display. When learning is more difficult, in contrast, consumers with lower search costs instead benefit from transparent display, while those with higher search cost continue to be harmed. Extending or shifting transparency requirements to algorithms that recommend products or provide advice reverts these welfare changes: the consumer's payoff is as if the algorithm is opaque. As a result, the welfare consequences of transparency requirements are nuanced, and a complete analysis should account for both the steering and persuasion effects.

\newpage

\setstretch{1}
\bibliographystyle{ARX}
\bibliography{ARX}

@article{A2007,
  author  = {Arbatskaya, Maria},
  title   = {{Ordered Search}},
  journal = {The RAND Journal of Economics},
  volume  = {38},
  number  = {1},
  pages   = {119--126},
  year    = {2007},
  doi     = {10.1111/j.1756-2171.2007.tb00047.x}
}

@article{Armstrong_Zhou_2011_EJ,
  author = {Armstrong, Mark and Zhou, Jidong},
  title = {{Paying for Prominence}},
  journal = {The Economic Journal},
  volume = {121},
  number = {556},
  pages = {F368--F395},
  year = {2011},
  doi = {10.1111/j.1468-0297.2011.02469.x}
}

@article{Athey_Ellison_2011_QJE,
  author = {Athey, Susan and Ellison, Glenn},
  title = {{Position Auctions with Consumer Search}},
  journal = {The Quarterly Journal of Economics},
  volume = {126},
  number = {3},
  pages = {1213--1270},
  year = {2011},
  doi = {10.1093/qje/qjr028}
}

@article{AVZ2009,
  author  = {Armstrong, Mark and Vickers, John and Zhou, Jidong},
  title   = {{Prominence and Consumer Search}},
  journal = {The RAND Journal of Economics},
  volume  = {40},
  number  = {2},
  pages   = {209--233},
  year    = {2009},
  doi     = {10.1111/j.1756-2171.2009.00062.x}
}

@article{BarIsaac_Shelegia_2025_JPEMicro,
  author  = {Bar-Isaac, Heski and Shelegia, Sandro},
  title   = {{Monetizing Steering}},
  journal = {Journal of Political Economy Microeconomics},
  year    = {2025},
  doi     = {10.1086/740100},
  note    = {Just accepted}
}

@unpublished{BS2025,
  author = {Boleslavsky, Raphael and Shadmehr, Mehdi},
  title  = {Signaling With Commitment},
  year   = {2025},
  note   = {Working paper}
}

@article{CH2011,
  author  = {Chen, Yongmin and He, Chuan},
  title   = {{Paid Placement: Advertising and Search on the Internet}},
  journal = {The Economic Journal},
  volume  = {121},
  number  = {556},
  pages   = {F309--F328},
  year    = {2011},
  doi     = {10.1111/j.1468-0297.2011.02466.x}
}

@article{deCorniere_Taylor_2019_RAND,
  author  = {De Corni\`ere, Alexandre and Taylor, Greg},
  title   = {A Model of Biased Intermediation},
  journal = {RAND Journal of Economics},
  volume  = {50},
  number  = {4},
  pages   = {854--882},
  year    = {2019},
  doi     = {10.1111/1756-2171.12298}
}

@misc{EU_2022_DSA,
  author = {{European Union}},
  title = {{Regulation (EU) 2022/2065 of the European Parliament and of the Council of 19 October 2022 on a Single Market for Digital Services and amending Directive 2000/31/EC (Digital Services Act)}},
  howpublished = {Official Journal of the European Union, L 277, 27 October 2022, pp. 1--102},
  year = {2022},
  url = {https://eur-lex.europa.eu/eli/reg/2022/2065/oj}
}

@article{BB2024,
Author = {Bergemann, Dirk and Bonatti, Alessandro},
Title = {Data, Competition, and Digital Platforms},
Journal = {American Economic Review},
Volume = {114},
Number = {8},
Year = {2024},
Month = {August},
Pages = {2553–95},
DOI = {10.1257/aer.20230478},
URL = {https://www.aeaweb.org/articles?id=10.1257/aer.20230478}}

@article{IO2009,
Author = {Inderst, Roman and Ottaviani, Marco},
Title = {Misselling through Agents},
Journal = {American Economic Review},
Volume = {99},
Number = {3},
Year = {2009},
Month = {June},
Pages = {883–908},
DOI = {10.1257/aer.99.3.883},
URL = {https://www.aeaweb.org/articles?id=10.1257/aer.99.3.883}}

@article{Ghose_Ipeirotis_Li_2014_MgmtSci,
  author  = {Ghose, Anindya and Ipeirotis, Panagiotis G. and Li, Beibei},
  title   = {Examining the Impact of Ranking on Consumer Behavior and Search Engine Revenue},
  journal = {Management Science},
  volume  = {60},
  number  = {7},
  pages   = {1632--1654},
  year    = {2014},
  doi     = {10.1287/mnsc.2013.1828}
}

@article{Hagiu_Jullien_2011_RAND,
  author  = {Hagiu, Andrei and Jullien, Bruno},
  title   = {{Why Do Intermediaries Divert Search?}},
  journal = {The RAND Journal of Economics},
  volume  = {42},
  number  = {2},
  pages   = {337--362},
  year    = {2011},
  doi     = {10.1111/j.1756-2171.2011.00136.x}
}

@article{Janssen_et_al_2026_EJ,
  author  = {Janssen, Maarten C. W. and Jungbauer, Thomas and Preuss, Marcel and Williams, Cole},
  title   = {{Search Platforms: Big Data and Sponsored Positions}},
  journal = {The Economic Journal},
  year    = {2026},
  pages   = {ueag002},
  doi     = {10.1093/ej/ueag002},
  note    = {Published online January 13, 2026}
}

@article{Janssen_et_al_2026_IJIO,
  author  = {Janssen, Maarten C. W. and Jungbauer, Thomas and Preuss, Marcel and Williams, Cole},
  title   = {{Optimally Informative Rankings and Consumer Search}},
  journal = {International Journal of Industrial Organization},
  year    = {2026},
  pages   = {103282},
  doi     = {10.1016/j.ijindorg.2026.103282},
  note    = {Available online March 7, 2026}
}

@article{Kamenica_Gentzkow_2011_AER,
  author = {Kamenica, Emir and Gentzkow, Matthew},
  title = {{Bayesian Persuasion}},
  journal = {American Economic Review},
  volume = {101},
  number = {6},
  pages = {2590--2615},
  year = {2011}
}

@article{Nocke_Rey_2024_JPE,
  author  = {Nocke, Volker and Rey, Patrick},
  title   = {{Consumer Search, Steering, and Choice Overload}},
  journal = {Journal of Political Economy},
  volume  = {132},
  number  = {5},
  pages   = {1684--1739},
  year    = {2024},
  doi     = {10.1086/728108}
}

@misc{S2019,
  author       = {{Stigler Committee}},
  title        = {Stigler Committee on Digital Platforms: Final Report},
  year         = {2019},
  institution  = {Stigler Center, University of Chicago Booth School of Business},
  note         = {Released September 16, 2019}
}

@article{Sahni_Nair_2020_RESTUD,
  title={Does Advertising Serve as a Signal? Evidence from a Field Experiment in Mobile Search},
  author={Sahni, Navdeep S. and Nair, Harikesh S.},
  journal={Review of Economic Studies},
  volume={87},
  number={3},
  pages={1529--1564},
  year={2020},
}

@article{Teh_Wright_2022_AEJMicro,
  author  = {Teh, Tat-How and Wright, Julian},
  title   = {{Intermediation and Steering: Competition in Prices and Commissions}},
  journal = {American Economic Journal: Microeconomics},
  volume  = {14},
  number  = {2},
  pages   = {281--321},
  year    = {2022},
  doi     = {10.1257/mic.20190344}
}

@article{Ursu_2018_MktgSci,
  author  = {Ursu, Raluca M.},
  title   = {{The Power of Rankings: Quantifying the Effect of Rankings on Online Consumer Search and Purchase Decisions}},
  journal = {Marketing Science},
  volume  = {37},
  number  = {4},
  pages   = {530--552},
  year    = {2018},
  doi     = {10.1287/mksc.2017.1072}
}

@unpublished{Yu2025WelfareSponsored,
  author = {Yu, Chuan},
  title  = {The Welfare Effects of Sponsored Product Advertising},
  year   = {2025},
  note   = {Working paper}
}
\setstretch{1.2}

\begin{appendices}
\section*{Appendix}
\subsection*{Proofs for Section \ref{sec:opacity} (Opacity)}
\begin{proof}[Proof of Proposition \ref{prop:equilibrium under opacity}] From the text, \[\Delta_1\equiv z_A^1-z_B^1=f+r(1-b)+(1-r)(a-b)\qquad\Delta_0\equiv z_A^0-z_B^0=f+(1-r)(a-b),\]
and
\[a\in\multi\{\widehat{\mu}_A\geq\theta_A\} \qquad b\in\multi\{\widehat{\mu}_B\geq\theta_B\}\qquad \pi_\omega=\multi\{\Delta_\omega\geq 0\}.\] Furthermore, $(\widehat{\mu}_A,\widehat{\mu}_B)$ are consistent with Bayes' rule applied to the conjectured algorithm $\widehat\Pi$, and the equilibrium conjecture is correct $\widehat{\Pi}=\Pi$.

\textit{A-always.} In an $A$-always equilibrium, $\widehat{\pi}_\omega=\pi_\omega=1$. Consequently, $\mu_A=\prior$ and $\mu_B$ is off-path. Such an equilibrium exists at prior $\prior$ if and only if a belief $\mu_B\in[0,1]$ can be found such that $\Delta_\omega\geq 0$ for $\omega\in\{0,1\}$. Consider $\mu_B=\prior$. Because $\theta_A<\theta_B$, we have $a\geq b$ for any possible selections. Therefore, $\Delta_\omega\geq f$ for $\omega\in\{0,1\}$, and such an equilibrium exists at every prior $\prior$. Other choices of $\mu_B$ also sustain the equilibrium. 

Using the expression in the text we have that the possible profits in an $A$-always equilibrium at prior $\prior$ are a correspondence, $\overline{v}_A(\prior)=f+r\mu+(1-r)\multi\{\prior\geq\theta_A\}$. This correspondence arises because  any $a\in[0,1]$ is possible at $\prior=\theta_A$.

\textit{B-always.} In a $B$-always equilibrium, $\widehat{\pi}_\omega=\pi_\omega=0$. Consequently, $\mu_B=\prior$ and $\mu_A$ is off-path. Such an equilibrium exists at prior $\prior$ if and only if a belief $\mu_A\in[0,1]$ can be found such that $\Delta_\omega\leq 0$ for $\omega\in\{0,1\}$. Because $\Delta_\omega$ is increasing in $a$, a $B$-always equilibrium exists only if $\Delta_\omega\leq 0$ with $a=0$. Furthermore, any off-path belief $\mu_A\leq\theta_A$ sustains $a=0$. Hence, a class of $B$-only equilibria exists if and only if $a=0\Rightarrow\Delta_\omega\leq 0$. Noting that $a=0$ implies $\Delta_1=f+r-b$ and $\Delta_0=f-(1-r)b$, such an equilibrium exists if and only if $b\geq \overline{b}\equiv\max\{f+r,\tfrac{f}{1-r}\}=f+r$, where the last equality follows from regularity condition (ii) ( $f<1-r$). With $f+r<1$, a sufficient condition for existence of $B$-only equilibria is $\mu_B=\prior>\theta_B$, so that $b=1$. Furthermore, with $\mu_B=\prior=\theta_B$, any $b\in[0,1]$ can be selected, and therefore $b\geq f+r$ is  necessary and sufficient for existence of $B$-only equilibria. With $\mu_B<\theta_B$, the only possible selection is $b=0<f+r$, and therefore a $B$-only equilibrium does not exist.

Using the expression in the text we have that the profit in a $B$-always equilibrium at prior $\prior>\theta_B$ is $v_B(\prior)=r\prior+1-r$. At $\prior=\theta_B$, multiple profits are possible because of the \searcher{}'s mixing, but $v_B(\theta_B)=r\theta_B+(1-r)$ is an upper bound, corresponding to $b=1$.

\textit{Split Equilibrium.} Consider a possible split equilibrium in which $\mu_A\neq\prior$ and $\mu_B\neq\prior$. 

\underline{Claim 1}: If a split equilibrium exists, then $\mu_A\leq\theta_A$. 

Suppose instead that $\mu_A>\theta_A$. It follows that $a=1$. Hence, $\Delta_1=f+1-b>0$ and $\Delta_0=f+(1-r)(1-b)>0$. By implication, $\widehat{\pi}_\omega=\pi_\omega=1$, and thus $\mu_A=\prior$, contradicting $\mu_A\neq\prior$.

\underline{Claim 2}: If a split equilibrium exists, then $\mu_A\geq\theta_A$. 

Suppose instead that $\mu_A<\theta_A$. It follows that $a=0$. Hence, $\Delta_1=f+r-b$ and $\Delta_0=f-(1-r)b$. 

First, suppose that $\mu_B<\theta_B$, so that $b=0$. In this case, $\Delta_1>\Delta_0=f>0$, which implies $\pi_0=\pi_1=1$. It follows that $\mu_A=\prior$,  contradicting $\mu_A\neq\prior$. 

Second, suppose that $\mu_B=\theta_B$, so that any $b\in[0,1]$ can be selected. In this case, $\Delta_1=f+r-b\geq\Delta_0=f-(1-r)b$. Note that $\mu_A<\theta_A$ requires $\pi_0>0$. In turn, $\pi_0>0$ implies $\Delta_0=0$, which implies $b=f/(1-r)$, and in turn $\Delta_1>0$ and $\pi_1=1$. Thus, there are two possibilities (1) $\pi_1=\pi_0=1$ and (2) $\pi_1=1$ and $\pi_0\in(0,1)$. Possibility (1) implies $\mu_A=\prior$, violating the split structure. Possibility (2) and Bayes' rule imply $\mu_B=0$, contradicting $\mu_B=\theta_B$.

Third, suppose that $\mu_B>\theta_B$, so that $b=1$. In this case, $\Delta_1=\Delta_0=f+r-1<0$. It follows that $\pi_1=\pi_0=0$, which implies $\mu_B=\prior$, violating the split structure, $\mu_B\neq\prior$. 

\underline{Claim 3}: If a split equilibrium exists, then $\mu_A=\theta_A$ and $\mu_B\geq \theta_B$. 

Combining Claims 1 and 2, we have $\mu_A=\theta_A$, so that any $a\in[0,1]$ is possible. Suppose to the contrary that $\mu_B<\theta_B$, so that $b=0$. It follows that $\Delta_1=f+r+(1-r)a>\Delta_0=f+(1-r)a>0$, and hence $\pi_0=\pi_1=1$, contradicting the split structure.

\underline{Claim 4}: A split equilibrium does not exist if $\prior\leq \theta_A$. 

Claim 4 follows from Claim 3 and the law of iterated expectations, which requires that the prior lies between the posteriors. Furthermore, at $\prior=\theta_A$, Claim 3 implies $\mu_A=\prior$, violating the split structure.

\underline{Claim 5}: If a split equilibrium exists then $\prior<\mu_B$. 

Claim 5 follows from Claim 3 and the law of iterated expectations.

\underline{Claim 6}: If a split equilibrium exists with $\mu_B=m$, then $m\in[\theta_B,1]$ and $\prior\in(\theta_A,m)$.

Condition $\prior\in(\theta_A,m)$ follows from Claims 4 and 5. Condition $m\in[\theta_B,1]$ follows from Claim 3. 

\underline{Claim 7}: Consider any $m\in[\theta_B,1]$ and $\prior\in(\theta_A,m)$. A split equilibrium such that $\mu_B=m$ exists.

Because $\mu_B=m\geq\theta_B$, it is always possible to select $b=1$. Furthermore, because $\mu_A=\theta_A$, it is possible to select $a=1-f/(1-r)$. With such $\{a,b\}$, we have $\Delta_1=\Delta_0=0$, and therefore any algorithm is a best response for the platform, $\pi_1\in[0,1]$ and $\pi_0\in[0,1]$. The required equilibrium exists if and only if an algorithm can be found induces posterior beliefs $\mu_A=\theta_A$ and $\mu_B=m$. Such an algorithm exists because the prior lies strictly between the posteriors, $\theta_A<\prior<m=\mu_B$.

\underline{Claim 8}: Consider any $\mu_1\in(\theta_A,1)$. A continuum of split equilibra exists. In any such equilibrium, prominence reveals bad news, $\mu_A<\prior<\mu_B$.

We have shown in Claim 7 for $m\in[\theta_B,1]$ and $\theta_A<\prior<m$, a split equilibrium exists with the property that $\mu_B=m$. Varying $m$ over the non-empty interval $\{(\prior,1]\cap[\theta_B,1]\}$ generates a continuum of equilibria. That prominence reveals bad news follows immediately from $\theta_A<\prior<m$. 

\underline{Claim 9}: In every split equilibrium $a=1-f/(1-r)$,  $b=1$, and $\Delta_1=\Delta_0=0$. 

From Claim 3, in equilibrium $\mu_A=\theta_A$ and $\mu_B\geq \theta_B$, and from Claim 5, $\prior\in(\theta_A,\mu_B)$. Using Bayes rule,  $\mu_A=\theta_A<\prior$ implies $0<\pi_1<\pi_0$. First note that $\pi_1<\pi_0$, which implies $\Delta_1\leq\Delta_0$. Observing that $\Delta_1=\Delta_0+r(1-b)$, we have $b=1$. Second, note that $0<\pi_1<\pi_0\leq 1$ implies $\pi_1\in(0,1)$, and hence $\Delta_1=0$. Together with $b=1$, we have $a=1-f/(1-r)$. Finally, with these $a=1-f/(1-r)$ and $b=1$, we have $\Delta_0=0$.

\underline{Claim 10:} Consider any $\mu_1\in(\theta_A,1)$. All split equilibria that exist at this prior are payoff-equivalent, $v_S=r\prior+(1-r)$.

Using the expression for profit in the text, and noting $\Delta_\omega=0$ and $b=1$, we have that profit in a split equilibrium is $v_S=\prior+(1-\prior)(1-r)=r\prior+(1-r)$.

\textit{Payoff comparison.} Note that when $\prior>\theta_A$, split equilibria, $B$-always equilibria, or both exist. Whenever both exist, the split equilibrium is a weak upper bound on the payoff of the $B$-always equilibrium, i.e. $v_S=r\mu+1-r\geq v_B$. At such priors, the profit if the $A$-only equilibrium is $v_A(\prior)=f+r\mu+(1-r)$. The ranking is immediate.
\end{proof}

\subsection*{Proofs for Section \ref{sec:transparency} (Transparency)}
\begin{proof}[Proof of Lemma \ref{lem:transparent profit}]
Suppose that algorithm designs a transparent algorithm $\Pi$ that induces posterior beliefs $(\mu_A,\mu_B)$. Let $a(\mu_A,\mu_B)$ be the probability that the \searcher{} stops when $A$ is prominent and the signal realization is inconclusive ($R=\phi$) and $b(\mu_A,\mu_B)$ be the probability that she continues when $B$ is prominent and the signal realization is inconclusive ($R=\phi$). In equilibrium, these must be consistent with Lemma \ref{lem:searcher strategy},  \begin{equation}\label{eq:ab}a(\mu_A,\mu_B)\in\multi\{\mu_A\geq\theta_A\}\qquad\qquad b(\mu_A,\mu_B)\in\multi\{\mu_B\geq\theta_B\}\end{equation}
We construct the platform's expected profit correspondence with a transparent algorithm. Following the same steps as in \eqref{eq:opaque profit}, we have
\begin{align*}
\prior&\bigl[\pi_1\{f+r+(1-r)a(\mu_A,\mu_B)\}+(1-\pi_1)b(\mu_A,\mu_B)]\\
+(1-\prior)&\bigl[\pi_0\{f+(1-r)a(\mu_A,\mu_B)\}+(1-\pi_0)(1-r)b(\mu_A,\mu_B)\bigr].
\end{align*}

We proceed in two steps. First, we reduce this expression to depend only on the posterior beliefs. Second, we show that in equilibrium, tie-breaking must be platform-preferred.

\textit{Step 1}. We apply the belief-based approach to eliminate $(\pi_0,\pi_1)$. Regrouping terms, we have,
\[
\bigl[\prior\pi_1+(1-\prior)\pi_0\bigr]\bigl[f+(1-r)a(\mu_A,\mu_B)\bigr]
+\prior\pi_1 r
+\bigl[\prior(1-\pi_1)+(1-\prior)(1-\pi_0)(1-r)\bigr]b(\mu_A,\mu_B).
\]
Substituting \eqref{eq:BBA}, we have 
\begin{align}\label{eq:appendixprofit}
&\tau_A\bigl[f+(1-r)a(\mu_A,\mu_B)\bigr]
+r\tau_A\mu_A
+\bigl[\tau_B\mu_B+\tau_B(1-\mu_B)(1-r)]b(\mu_A,\mu_B)=\nonumber\\
&\tau_A[f+r\mu_A+(1-r)a(\mu_A,\mu_B)\bigr]+\tau_B[r\mu_B+1-r\bigr]b(\mu_A,\mu_B)
\end{align}
\textit{Step 2.} We show that tie-breaking must be platform-preferred. In particular, if $\mu_A=\theta_A$ in equilibrium, then  $a(\theta_A,\mu_B)=1$. Similarly, if $\mu_B=\theta_B$ in equilibrium, then $b(\mu_A,\theta_B)=1$. We establish the result for $a(\cdot,\cdot)$ ($b(\cdot,\cdot)$ is analogous).

Consider a possible equilibrium in which $\mu_A=\theta_A$ and $\tau_A>0$. We show that if $a(\theta_A,\mu_B)<1$, then the platform has a strictly beneficial deviation.

Case 1: $\prior\neq\theta_A$. It follows that $\mu_A\neq\mu_B$. Because $\prior\neq\theta_A$, a marginal deviation from $\mu_A=\theta_A$ to $\mu'_A=\theta_A^+$ is feasible. Such a deviation has only a marginal effect on the probabilities $(\tau_A,\tau_B)$ but generates a discontinuous increase in the profit at realization $\mu'_A$. Thus, the change in profit from such a deviation is $\tau_A(1-r)(1-a(\theta_A,\mu_B))>0$. Thus, such a deviation is strictly profitable. 

Case 2: $\mu_A=\theta_A=\prior$. The law of iterated expectations implies either (i) $\tau_A=1$ or (ii) $\mu_B=\theta_A$ and $\tau_A\in(0,1)$. 

Suppose (i). If $\mu_A=\theta_A$ and $\tau_A=1$, then the platform's payoff is $f+r\theta_A+(1-r)a(\theta_A,\mu_B)$. Consider the following marginal deviation: $(\mu_B=0,\mu'_A=\theta_i^+)$. Note that the law of iterated expectations implies that with this deviation $\tau_A=\prior/\theta_i^+=\theta_A/\theta_i^+\approx 1$. Furthermore, with $\mu'_A=\theta_A^+>\theta_A$, we have $a(\mu_A',0)=1$. Hence, the platform's expected profit from such a deviation is \[\frac{\theta_A}{\theta_A^+}(f+r\theta_A^++(1-r)]\approx f+r\theta_A+1-r>f+r\theta_A+(1-r)a(\theta_A,\mu_B).\] Thus, such a deviation is strictly profitable.

Suppose (ii): $\mu_A=\mu_B=\prior=\theta_A$ and $\tau_A\tau_B>0$. The platform payoff is therefore,
\[\tau_A[f+r\mu_A+(1-r)a(\theta_A,\theta_A)\bigr],\] where use has been made of $b(\theta_A,\theta_A)=0$. Consider a deviation to a new algorithm $(\mu'_A=\theta_A+\tau_B\epsilon,\mu'_B=\theta_A-\tau_A\epsilon)$, for some small $\epsilon>0$. For the new algorithm, the law of iterated expectations implies that each product is prominent with the same probability as in the original algorithm, i.e.
 \(\tau_A\mu_A'+\tau_B\mu_B'=\tau_A(\theta_A+\tau_B\epsilon)+\tau_B(\theta_A-\tau_A\epsilon)=\theta_A.\) Furthermore, with $\mu_A'>\theta_A$, we have $a(\mu_A',\mu_B')=1$ and $b(\mu_A',\mu_B')=0$.
 Hence, the platform's expected payoff from the deviation is \[\tau_A[f+r\mu_A'+(1-r)a(\mu_A',\mu_B')\bigr]=\tau_A[f+r\mu_A'+(1-r)\bigr]>\tau_A[f+r\mu_A+(1-r)a(\theta_A,\theta_A)\bigr].\]
Thus, such a deviation is strictly profitable.

\textit{Step 3.} Substituting $a(\theta_A,\mu_B)=1$ and $b(\mu_A,\theta_B)=1$ into \eqref{eq:ab} implies \(a(\mu_A,\mu_B)=\mathbbm{1}\{\mu_A\geq\theta_A\}\) and \(b(\mu_A,\mu_B)=\mathbbm{1}\{\mu_B\geq\theta_B\}\). Substituting into \eqref{eq:appendixprofit} completes the proof.
\end{proof}

\begin{proof}[Proof of Proposition \ref{prop:transparency and welfare}]
The platform's payoff ranking follows  from Proposition \ref{prop:equilibrium under transparency}. 

We show that both consumer types are better off under transparency at priors $\prior\in(\widehat{\mu},\theta_A)$, where the platform uses positive sorting in the equilibrium under transparency. Consider a consumer of type $j\in\{se,cap\}$. The consumer's payoff from positive sorting  is $\tau(\mu_A^+)u^j_A(\mu_A^+)+\tau(\mu_B^+)u^j_B(\mu_B^+)$ for $j\in\{se,cap\}$. Furthermore, linearity of $u_A^{j}(\cdot)$ on $[0,\theta_A]$ implies that $u_A^{j}(\prior)=\tau(\mu_A^+)u^j_A(\mu_A^+)+\tau(\mu_B^+)u^j_A(\mu_B^+)$. Subtracting, the payoff with positive sorting is higher whenever $u^j_B(\mu_B^+)>u^j_A(\mu_B^+)$. For the \searcher{}, $u^{se}_B(\mu_B^+)=\out$ and $u^{se}_A(\mu_B^+)=\out-c$. For the \captive{}, $u^{cap}_B(\mu_B^+)=\out$ and $u_A(\mu_B^+)=0$. The result follows.

We show that the \captive{} is worse off under transparency at low priors, $\prior\in(0,\widehat{\mu})$ where the platform uses negative sorting in the equilibrium under transparency. As in the previous step, negative sorting is worse than $A$-only whenever $u_B^{cap}(\mu_B^-)<u^{cap}_A(\mu_B^-)$. Observing that $u_B^{cap}(\mu_B^-)=\out$, $u_A^{cap}(\mu_B^-)=\theta_B$ and $\theta_B>\out$, completes the proof.

We show that the \searcher{} can be better or worse off under transparency at low priors under the regularity conditions. From Figure \ref{fig:transparency and welfare}, the payoff of negative sorting is the line interpolating $(0,\out-c)$ and $(\theta_B,\out)$, i.e. $L(\mu)=\out-c+\tfrac{c}{\theta_B}\mu$. At low prior beliefs $(\mu<\widehat\mu)$, transparency is better if and only if $L(\theta_A)>u_A^{se}(\theta_A)=\out-c+r(\out-c)(1-\out+c)$, equivalently
\[Q(r)\equiv c-\bigl[\out(1-\out)+c\out+c+c^2\bigr]r+\bigl[\out(1-\out)+\out c\bigr]r^2>0.\] 
First, fix any $(c,\out,f)$ satisfying regularity condition (i) and $f<1$. If $r$ is sufficiently small, then all three regularity conditions hold. Furthermore, for $r$ sufficiently small, the required inequality reduces to $c>0$. Therefore, for each such $(c,\out,f)$, the \searcher{} benefits from transparency if $r$ is sufficiently small.

Next, fix any $(c,\out)$ satisfying regularity condition (i) and assume that $f< \bar{f}=\tfrac{c}{1+c}$. Recall that the remaining regularity conditions (ii) and (iii), require $r\leq\bar{r}\equiv\min\{1-f,\tfrac{c}{\out(1-\out)}\}$. Next, note that 
\begin{align*}
Q(0)=c>0\\
Q(\frac{c}{\out(1-\out)})=-\frac{c^2}{(1-\out)^2}(1-\out-c)<0\\
Q(1-f)=-c^2+f\bigl[c(1+c)-\out(1-\out+c)\bigr]+f^2\out(1-\out+c)
\leq -c^2+fc(1+c)<0
\end{align*}
Thus, $Q(0)>0$ and $Q(\bar{r})<0$. Because $Q(\cdot)$ is a convex parabola, it therefore has exactly one root $r^*(c,\out)$ in the interval $(0,\bar{r})$. Furthermore, at this root, $Q(\cdot)$ switches from positive to negative. Thus, if $(c,\out)$ satisfy regularity condition (i) and $f<\bar{f}$, then for $r<r^*(c,\out)$ the \searcher{} benefits from transparency, and for $r^*(c,\out)<r<\bar{r}$, the \searcher{} is harmed by transparency.\end{proof}

\subsection*{Proofs for Section \ref{sec:monetizing} (Monetizing Prominence)}

\begin{proof}[\textbf{Proof of \autoref{lem:competitive commissions}}]
Let \(i\) have a higher effective value, \(\eta_i>\eta_{j}\). We refer to $i$ as strong and ${j}$ as weak.

\textit{Step 1: structure of equilibrium.} No pure-strategy equilibrium exists. If one firm bids strictly more than the other, it can reduce its bid slightly and still remain the higher-commission product, If both bids coincide, one firm can move slightly above or below the common bid depending on whether winning at that bid is profitable. Hence any equilibrium is mixed.

Because becoming the high-commission firm raises sales from \(p_i^L\) to \(p_i^H\), firm \(i\)'s expected payoff from bid \(k\) against rival cdf \(G_{j}\) is
\[
\pi_i(k)=(m-k)\bigl[p_i^L+G_{j}(k)\Delta\bigr].
\]
The highest bid weak firm \(j\) is willing to submit solves the indifference condition between winning with bid \(k\) and losing with bid \(0\):
\[
(m-k)p_{j}^H=mp_{j}^L,
\]
so the upper edge of support is \(k=\eta_{j}\). The stronger bidder never bids above \(\eta_{j}\) because the weaker bidder never does.

\textit{Step 2: the stronger bidder's cdf.} The weaker bidder's equilibrium payoff is its payoff from bidding \(0\) and losing,
\[
U_{j}=mp_{j}^L.
\]
Indifference on the support therefore requires
\[
(m-k)\bigl[p_{j}^L+G_i(k)\Delta_{j}\bigr]=mp_{j}^L,
\]
which rearranges to
\[
G_i(k)=\frac{k\,p_{j}^L}{(m-k)\Delta},
\qquad k\in[0,\eta_j].
\]
At \(k=\eta_{j}\) this cdf reaches one.

\textit{Step 3: the weaker bidder's cdf and atom.} Let \(a\) denote the atom that weak bidder \(j\) places at \(0\). The stronger bidder's equilibrium payoff from bidding \(0\) is then
\[
U_i=m\bigl[p_i^L+a\Delta\bigr].
\]
Since weak bidder \(j\)'s cdf reaches one at \(\eta_{j}\), stronger bidder \(i\) must also be indifferent between bidding \(0\) and bidding \(\eta_{j}\), which implies
\[
U_i=(m-\eta_{j})p_i^H.
\]
Equating the two expressions yields
\[
a=1-\frac{\eta_{j}}{\eta_i}.
\]
Indifference of the stronger bidder on the support gives
\[
(m-k)\bigl[p_i^L+G_{j}(k)\Delta\bigr]=(m-\eta_{j})p_i^H,
\]
so
\[
G_{j}(k)=\frac{(m-\eta_{j})p_i^H}{(m-k)\Delta}-\frac{p_i^L}{\Delta},
\qquad k\in(0,\eta_{j}].
\]
At \(k=0\) this equals \(1-\eta_{j}/\eta_i\), the atom derived above, and at \(k=\eta_{j}\) it equals one.

\textit{Step 4: uniqueness.} The upper support \(\eta_{j}\) is uniquely pinned down by the weaker bidder's indifference condition. Given that support, the cdfs are uniquely pinned down by indifference of the two bidders over the common support. Hence the mixed equilibrium is unique. If \(\eta_A=\eta_B\), the same indifference conditions imply a common support \([0,\eta_A]\) and no atom at zero.

\textit{Step 5: rearrangement.} Substituting the identities $p_i^{H}/\Delta=m/\eta_i$ and $p_{j}^{H}/\Delta=m/\eta_{j}$ and simplifying delivers the expressions in the proposition.
\end{proof}

\begin{proof}[Proof of Lemma \ref{lem:favored treatment}]
Since the consumer always purchases a product,
\[
\Delta=p_A^H-p_A^L=p^H_A+p^H_B-1.
\]
\(A\) has a higher effective valuation if and only if \(p^H_A<p^H_B\leq 1\), and equilibrium commission revenue in this case is
\[
V(p^H_A,p^H_B)
=
m(p^H_A+p^H_B-1)\frac{p^H_A}{p^H_B}.
\]
Thus
\[
\frac{\partial V}{\partial p^H_A}
=
m\frac{2p^H_A+p^H_B-1}{p^H_B}
=
m\frac{p^H_A+\Delta}{p^H_B}>0\qquad\qquad \frac{\partial V}{\partial p^H_B}
=
m\frac{p^H_A(1-p^H_A)}{(p^H_B)^2}>0,
\]
where the last inequality is strict because $p_A^H<1$. The case where $B$ has higher effective valuation is symmetric.
\end{proof}
\end{appendices}
\end{document}